\documentclass[journal]{IEEEtran}

\usepackage{cite}
\usepackage{amsthm}
\usepackage{amsmath}
\usepackage{comment}
\usepackage{booktabs}
\usepackage{bm}
\usepackage{setspace}
\usepackage{graphicx}

\usepackage[usenames,dvipsnames]{color}
\usepackage{amsfonts,amssymb,amsmath}
\usepackage{subfigure}
\usepackage{enumerate}
\usepackage{cases}
\usepackage{array}
\usepackage{colortbl}
\usepackage{multirow}
\usepackage{arydshln}
\usepackage{soul, color, xcolor}
\usepackage{float}

\soulregister{\cite}7 
\soulregister{\citep}7 
\soulregister{\citet}7 
\soulregister{\ref}7 
\soulregister{\pageref}7 
\newtheorem{thm}{Theorem}
\newtheorem{cor}{Corollary}
\newtheorem{lem}{Lemma}
\newtheorem{prop}{Proposition}
\theoremstyle{remark}
\newtheorem{defn}{Definition}
\newtheorem{rem}{Remark}

\definecolor{Gray}{gray}{0.85}
\definecolor{mycyan}{cmyk}{.3,0,0,0}

\newcolumntype{a}{>{\columncolor{Gray}}c}
\newcolumntype{b}{>{\columncolor{white}}c}
\newcolumntype{d}{>{\columncolor{mycyan}}c}

\allowdisplaybreaks
\begin{document}
%
\title{\textcolor{black}{ Violation Probabilities of AoI and PAoI and Optimal Arrival Rate Allocation for the IoT-based Multi-Source Status Update System}}
\author{Tianci~Zhang, Shutong~Chen, Zhengchuan~Chen,~\IEEEmembership{Member,~IEEE,} Zhong~Tian,~\IEEEmembership{Member,~IEEE,} Yunjian~Jia,~\IEEEmembership{Member,~IEEE,} Min~Wang,~\IEEEmembership{Member,~IEEE,} and Dapeng~Oliver~Wu,~\IEEEmembership{Fellow,~IEEE}
\thanks{ T. Zhang, S. Chen, Z. Chen, Z. Tian and Y. Jia are with the School of Microelectronics and Communication Engineering, Chongqing University, Chongqing, China (E-mails: \{ztc, cst, czc, ztian, yunjian \}@cqu.edu.cn).
M. Wang is with the School of Optoelectronics Engineering, Chongqing University of Posts and Telecommunications, Chongqing, China. (E-mail:~wangm@cqupt.edu.cn).
D. O. Wu is with the Department of Computer Science, City University of HongKong, HongKong. (E-mail:~dpwu@ieee.org).}
}
%
\maketitle

\begin{abstract}
  Lots of real-time applications over Internet of things (IoT)-based status update systems have imperative demands on information freshness, which is usually evaluated by age of information (AoI).
  Compared to the average AoI and peak AoI (PAoI), violation probabilities and distributions of AoI and PAoI characterize the timeliness in more details.
  This paper studies the timeliness of the IoT-based multi-source status update system. \textcolor{black}{By modeling the system as a multi-source M/G/1/1 bufferless preemptive queue, general formulas of violation probabilities and probability density functions (p.d.f.s) of AoI and PAoI are derived with a time-domain approach.}
  For the case with negative-exponentially distributed service time, the violation probabilities and p.d.f.s are obtained in closed form. Moreover, the maximal violation probabilities of AoI and PAoI are proposed to characterize the overall timeliness. \textcolor{black}{To improve the overall timeliness under the resource constraint of  IoT-device, the arrival rate allocation scheme is used to minimize the  maximal violation probabilities.}
  It is proved that the optimal arrival rates can be found by convex optimization algorithms. In addition, it is obtained that the minimum of maximal violation probability of AoI (or PAoI) is achieved only if all violation probabilities of AoI (or PAoI) are equal.
  Finally, numerical results verify the theoretical analysis and show the effectiveness of the arrival rate allocation scheme.
\end{abstract}

\section{Introduction}
With the development of various sensing-communication technologies and cyber-physical control applications, Internet of things (IoT) has penetrated deeply into life.
Via various IoT devices and technologies,
the IoT-based systems can timely sample physical processes and achieve intelligent monitoring and interacting.
Lots of real-time applications over the IoT-based status update systems have emerged, such as the intelligent transport system, intelligent agriculture system and autonomous driving system \cite{8703738,9296855,9625017}.
Meanwhile, as an essential performance indicator of the real-time IoT-based status update systems, information freshness of the updates also receives widespread attention \cite{7397856}.
\textcolor{black}{The freshness-aware IoT network design, where IoT devices sense multiple physical processes and frequently update the status of these processes for a destination note, was synthetically studied in the insightful work \cite{8930830}.}
The timely information transmission provides bases to improve the overall performance of the IoT-based systems.

Age of information (AoI), as a receiver-centric indicator measuring the information freshness, characterizes the time elapsed since the latest successfully received update has been generated \cite{6195689}.  Compared to the traditional metrics, e.g., delay, AoI is more effective in characterizing information freshness \cite{8940930}.
\textcolor{black}{The differences between AoI and traditional metrics were also emphasised in \cite{8930830}.}
To analyze the AoI, status update systems are usually modeled as various queues.
There have been many works focusing on the AoI in the single- or multi-source queueing systems.
 \textcolor{black}{Kendall notation system is used to denote the queueing systems in our work.}
The authors of \cite{6195689} obtained the average AoIs for the single-source M/M/1, M/D/1 and D/M/1 first come first served (FCFS) queueing systems.
An exact expression for the average AoI in the multi-source M/M/1 FCFS queueing system was shown in \cite{9217386}.  The authors also derived three approximate expressions for the average AoI in the multi-source M/G/1 FCFS queueing system \cite{9099557}.

To improve the system timeliness, many packet management strategies have been proposed, such as the global- and self-preemption strategies \cite{9047958} in both single- and multi-source cases. {The global-preemption strategy is also known as preemption strategy.} The average AoI in the single-source M/M/1/1 preemptive system was derived in \cite{6310931}, revealing that the preemption achieves a lower average AoI than non-preemption. Moreover, the average AoI and peak AoI (PAoI) in the single-source M/G/1/1 system with preemption or non-preemption strategies were derived in \cite{8006504}.
\cite{8469047} presented the average AoI in the multi-source M/M/1/1 preemptive system. Additionally, they introduced stochastic hybrid systems (SHS) technique as a powerful tool to analyse the AoI.
\textcolor{black}{For the multi-source M/G/1/1 preemptive system, the average AoI and PAoI were derived  by the detour flow graph method \cite{8406928}.}

Apart from the packet management strategies, researchers have proposed various resources scheduling schemes \cite{8406928,9674782,8894836,9445676,23081055}, to further improve the overall system timeliness under the resource constraints. In particular, the arrival rate allocation scheme has attracted lots of attention.
For instance, it was found that setting the equal arrival rates for the sources can achieve the minimal total average AoI and average PAoI, under the equal priority condition \cite{8406928}. In \cite{23081055}, the authors considered the different timeliness requirements for different sources in the multi-source global- and self-preemptive systems. They obtained the optimal arrival rate allocations which minimize the weighted average AoIs. It was also shown that the global-preemption strategy outperforms the self-preemption strategy.

To characterize the timeliness, some works considered the distributions and violation probabilities of AoI and PAoI.
\cite{8820073} presented a general formula of distributions of AoI in single-source G/G/1 systems with FCFS or last come first served queueing disciplines.
\textcolor{black}{The moment generating functions (MGFs) of probability density functions (p.d.f.s) of AoI in single-source non-preemptive or preemptive in service/waiting systems with energy harvesting transmitter and negative-exponentially distributed service time were derived by using SHS \cite{9705518}.
Meanwhile, the MGFs in corresponding multi-source systems with energy harvesting transmitter were also derived by using SHS \cite{9732416}.
With the help of SHS, \cite{9312180} obtained the MGFs for dual-source self-preemptive or non-preemptive systems with negative-exponentially distributed service time.}
It is intuitive that the distributions describe the AoI process more comprehensively.
For instance, based on the distributions of AoI and PAoI, one can derive the average AoI and PAoI as well as variances of AoI and PAoI, which measure the stability of AoI process.
Moreover, \cite{9210029} studied the distribution of AoI in wireless networked control systems with two hops and minimized the violation probability of AoI, with respect to (w.r.t.) the arrival rates of updates.
For the single-source D/G/1 FCFS queuing system, the violation probability of PAoI was derived in \cite{8691802}.
The distributions and violation probabilities of AoI and PAoI were derived for the single-source M/M/1 and M/D/1 FCFS queueing systems in \cite{9324753}.
 \cite{8937801} achieved better ultra-reliable low-latency communications in vehicular networks, based on the violation probability of AoI.
Since the violation probability of AoI (or PAoI) characterizes the probability that AoI (or PAoI) exceeds a given AoI (or PAoI) threshold, the violation probabilities help evaluate the timeliness with strict age requirements.
\subsection{Motivations and Novelty}
As mentioned above, it is found that: a) The preemption strategy holds the potential of improving the timeliness; b) The arrival rate allocation scheme considering different timeliness requirements for sources can enhance the overall timeliness for the multi-source system; c) Compared to the average AoI and PAoI, the corresponding violation probabilities and distributions evaluate the timeliness more precisely and comprehensively, especially when AoI and PAoI fluctuate greatly. Besides, in many applications, when the AoI (or PAoI) exceeds a certain threshold, terrible results would outcome. The occurrence probability of outdated information, i.e. violation probability, must be well controlled.

Inspired by these, this work studies the violation probabilities and distributions of AoI and PAoI for an IoT-based multi-source status update system, where sensors update data to the same monitor through a transmitter allowing preemption. To the best of our knowledge, there is no work utilizing the violation probabilities and distributions of AoI and PAoI to both characterize and optimize the timeliness for multi-source preemptive systems. Compared to the single-source system, analysing the violation probabilities and distributions in the preemptive multi-source system is more sophisticated, since the preemption leads to disturbed queuing process of each stream and more involved inter-arrival time distribution. By modeling the system as a multi-source M/G/1/1 bufferless preemptive system, the violation probabilities and p.d.f.s of AoI and PAoI are derived.
\textcolor{black}{The Laplace transform (LT) of p.d.f. of the AoI was obtained by utilizing Palm calculus in \cite{9598815}.} \textcolor{black}{It is noteworthy that, as an alternative, we utilize a time-domain approach with more obvious physical meaning to analyse the violation probabilities and p.d.f.s of AoI and PAoI.}
\textcolor{black}{Specifically, we extend the insightful time-domain approach of \cite{8737474} to the multi-source case, and derive the violation probabilities based on several intermediate results in \cite{8406928}.
This provides a new thinking to study the violation probabilities and p.d.f.s for the multi-source systems.
In particular, for the M/M/1/1 case, the closed form expressions of the violation probabilities and p.d.f.s of AoI and PAoI are obtained. This presents more analytical results for practical system design, compared with the corresponding LTs and MGFs.
\textcolor{black}{Worthy of remark is that the MGF of AoI is usually calculated by using SHS for the system with negative-exponentially distributed service time \cite{9705518,9732416,9312180,9103131}, however not for the case with generally distributed service time.
In contrast, with the help of the time-domain approach, the violation probability and p.d.f. of AoI can be derived for the case with generally distributed service time.}
}
\textcolor{black}{To fully characterize the timeliness of the multi-source system, the maximal violation probability, i.e., the maximum among the violation probabilities of all streams, is proposed as an overall timeliness metric.
To improve the overall timeliness under the resource constraint of IoT device, we minimize the maximal violation probabilities of AoI and PAoI with the arrival rate allocation scheme, respectively.
Specifically, to get more insights, we focus on the arrival rate allocation in the M/M/1/1 case, as the negative-exponentially distributed service time with memoryless property is common for studying the timeliness of IoT-based systems.
However, it is noteworthy that, due to the complexity of objective functions w.r.t. arrival rates, it is difficult to solve the optimization problems and find the optimal arrival rate allocations.
With careful analyses on the violation probabilities, we prove that the formulated problems are convex and the optimum can be found by using standard convex optimization algorithms.
}

\subsection{Contributions and Organization}
The contributions of this work are summarized as follows:
\begin{itemize}
 \item By modeling the considered multi-source status update system as the multi-source M/G/1/1 bufferless preemptive system, we derive the general formulas of the violation probabilities and  p.d.f.s of AoI and PAoI. It is found that the AoI and inter-departure time follow the same distribution.
     Moreover, the system where the service time is negative-exponentially distributed is considered. The violation probabilities and p.d.f.s of AoI and PAoI are derived in closed form. Based on the obtained p.d.f.s, the variances of AoI and PAoI are further presented to characterize the stability of the AoI process.
  \item \textcolor{black}{For the case with negative-exponentially distributed service time, the maximal violation probabilities of AoI and PAoI are respectively minimized to achieve the optimal overall timeliness under the resource limitation of IoT device.} The arrival rate allocation scheme is utilized, and different timeliness requirements for sources are considered. Specifically, we formulate two optimization problems, and prove their convexity. The optimal arrival rates can be found by standard convex optimization algorithms. It is further obtained that the minimum of maximal violation probability of AoI (or PAoI) is achieved only if the violation probabilities of all sources are equal.
  \item According to the numerical results and simulations, our theoretical analyses are verified. It is found that the optimal arrival rate allocations based on the maximal violation probabilities can improve the overall timeliness remarkably.
      Moreover, it is shown that the timeliness of the system with negative-exponentially distributed service time, outperforms that of the system with deterministic or uniformly distributed service time.
\end{itemize}

The rest of this paper is organized as follows. Section \ref{sec 2} introduces the system model. Section \ref{sec 3} derives the violation probabilities of AoI and PAoI for the multi-source M/G/1/1 bufferless preemptive system.
Section \ref{sec 4} studies the violation probabilities for the system with negative-exponentially distributed service time.
In Section \ref{sec 5}, the overall timeliness is optimized based on the violation probabilities of AoI and PAoI, respectively.
Numerical results and simulations are shown in Section \ref{sec num}. Finally, Section \ref{sec con} concludes this paper.
\section{System Model}\label{sec 2}
\textcolor{black}{Consider an IoT-based  multi-source status update system, which consists of an IoT device and a remote monitor, as shown in Fig. \ref{fig:vvv}.}
\begin{figure}[!t]
\centering
\includegraphics[width=0.98\linewidth]{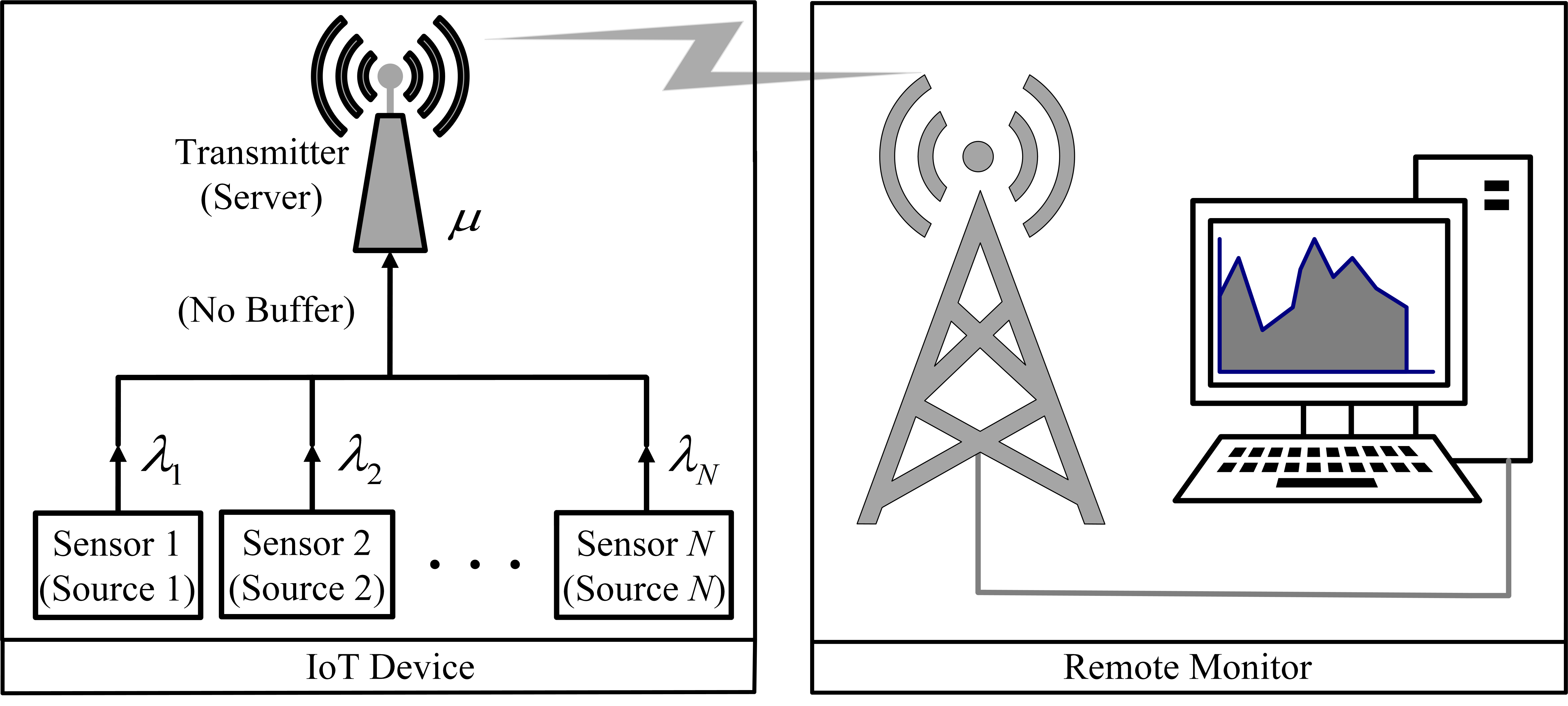}
\caption{The considered multi-source status update system.}
\label{fig:vvv}
\end{figure}
In the IoT device, there are $N$ independent sensors (sources) sensing different information, and a transmitter (server) which sends the collected information to the remote monitor.
The sensors are indexed by $i\in \mathcal{N}:=\left\{{{\rm{1,2,}}\cdots,N}\right\}$.
Specifically, the updates of sensor $i$, $\forall i\in \mathcal{N}$, are generated based on a Poisson process at arrival rate $\lambda_i$ and subsequently served by the transmitter at service rate $\mu$.  Additionally, we denote the total arrival rate of all sensors by $\lambda  = \sum\nolimits_{i = 1}^N {{\lambda _i}}$.
Since the consumable resources of the IoT device, such as the energy for sensing, are limited, it is assumed that the total arrival rate $\lambda$ is fixed.

To ensure that the received updates are as fresh as possible, the considered system is considered as a bufferless preemptive system. There is no buffer in the IoT device, and the updates of all sensors preempt each other to get served. Specifically, when the transmitter is busy and a new update arrives, the update under service will be discarded and the new one will be served instead immediately. The transmitter always serves the latest update and the monitor only receives the completely served updates. This reduces the system time of successfully received updates and further improve the freshness.
Accordingly, the considered system can be modeled as a multi-source M/G/1/1 bufferless preemptive system.
\begin{figure}[!t]
\centering
\includegraphics[width=1\linewidth]{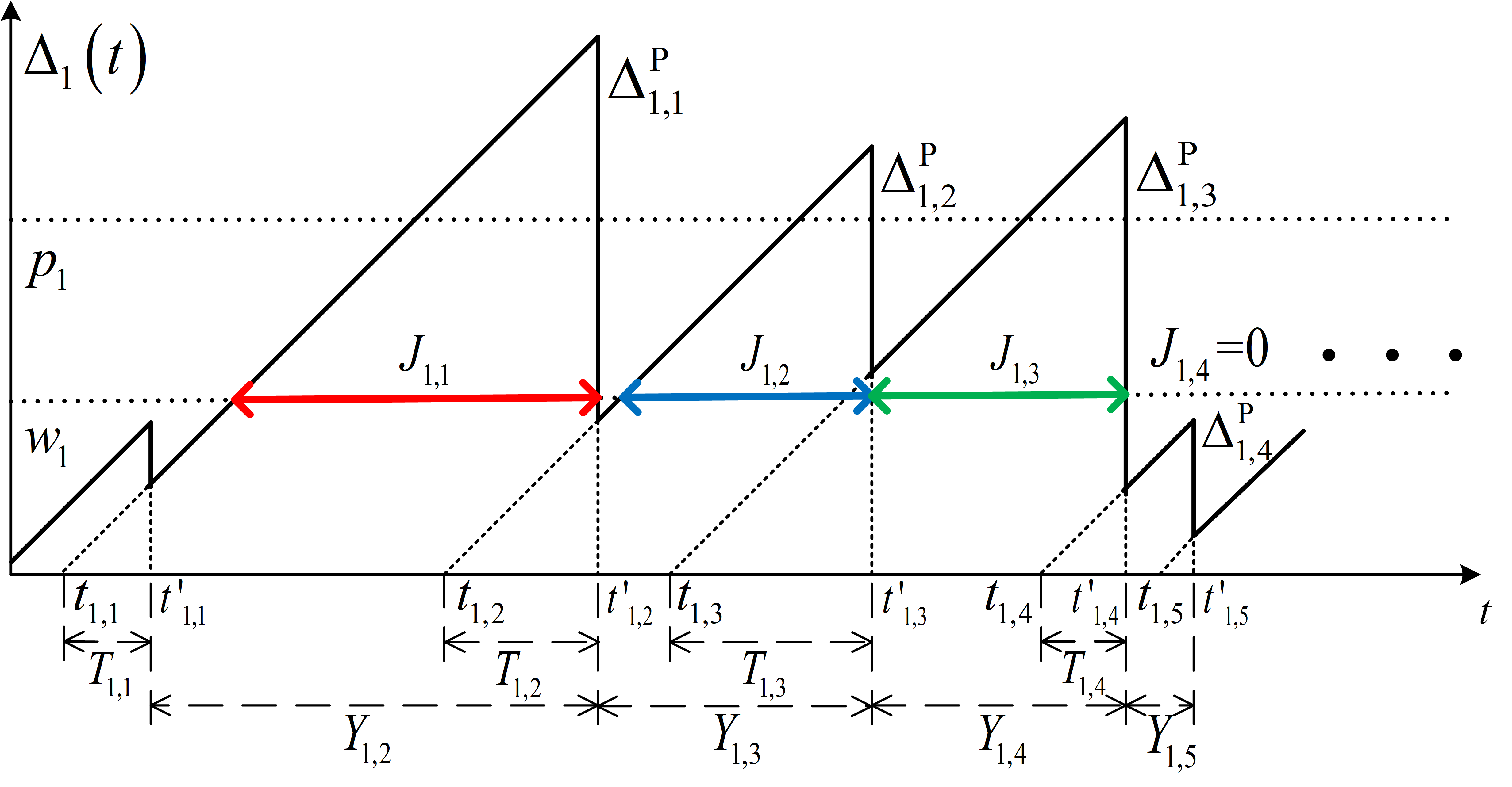}
\caption{Age evolution example for source 1.}
\label{fig:model}
\end{figure}

\textcolor{black}{
For clarity, we denote the arrival moment and departure moment, as well as system time and service time of the $k$-th successfully received update of source $i$, by $t_{i,k}$, $t'_{i,k}$, $T_{i,k}=t'_{i,k}-t_{i,k}$ and $S_{i,k}$, respectively, as shown in Fig. \ref{fig:model}. The inter-departure time between the $k$-th successfully received update from source $i$ and the $(k-1)$-th one is denoted by $Y_{i,k}=t'_{i,k}-t'_{i,k-1}$. Moreover, we denote the inter-arrival time between two consecutive updates from source $i$
by $X_i$.
Based on the system model, it is known that $X_i$ is negative-exponentially distributed with mean $1/{\lambda_i}$ and $S_{i,k}$ is independent identically distributed with mean $1/\mu$.}

%

Next, let us introduce the violation probabilities of AoI and PAoI to evaluate the timeliness performance.
For source $i$, the instantaneous AoI is $\Delta_{i}(t)=t-U_i(t)$,
where $n_{i}(t)=\max\{k|t'_{i,k}\leq t\}$ and $U_i(t)=t_{n_{i}(t)}$ are the index of the most recently successfully received update and its arrival moment, respectively.
For the $k$-th successfully received update of source $i$, the PAoI is the peak age in its AoI process, i.e., the time between the moment it arrives at the IoT device and the moment it is replaced by the $(k+1)$-th successfully received update of source $i$ at the monitor. The PAoI can be given by
\begin{align}\label{p2}
\Delta_{i,k}^{\text{P}}=Y_{i,k+1}+T_{i,k}.
\end{align}
Accordingly, the definitions of violation probabilities of AoI and PAoI are given in the following.
\begin{defn}
\begin{itshape}
 (Violation Probabilities of AoI and PAoI)
\end{itshape}
In the multi-source status update system, the violation probabilities of AoI and PAoI are defined as the probabilities that the AoI and PAoI of successfully received updates of source $i$ exceed given AoI and PAoI thresholds $w_i$ and $p_i$, respectively, i.e.,
\begin{align}\label{er}
\!P_{i}^{\text{A}}(w_i):=\Pr\{\Delta_{i}>w_i\} \; {\text{ and }} \; P_{i}^{\text{P}}(p_i):=\Pr\{\Delta_{i}^{\text{P}}>p_i\}\!.
\end{align}
\end{defn}
From (\ref{er}), it can be seen that the violation probability of AoI (or PAoI) indicates the occurrence frequency of AoI (or PAoI) greater than a preset AoI (PAoI) threshold.
\section{Violation Probabilities in the M/G/1/1 System}\label{sec 3}
In this section, the violation probabilities and p.d.f.s of PAoI and AoI are derived for the considered multi-source status update system, respectively.

First, in terms of long-term distribution, we omit the subscript and denote $S_{i,k}$, $T_{i,k}$ and $Y_{i,k}$ by $S$, $T$ and $Y_{i}$, respectively.
Let us denote the p.d.f.s of $S$, $T$ and $Y_{i}$ by $f_{S}(x)$, $f_{T}(x)$ and $f_{Y_{i}}(x)$, respectively.
The corresponding unilateral LTs is denoted by $L_{S}(s)$, $L_{T}(s)$ and $L_{Y_{i}}(s)$. We also use $L^{-1}[\cdot]$ to represent the inverse LT.
\subsection{Violation Probability of PAoI in the M/G/1/1 System}
In this subsection, the violation probability and p.d.f. of PAoI in the multi-source M/G/1/1 preemptive system are derived, as shown in the following theorem.
\begin{thm}\label{thm pdfp}
In the multi-source M/G/1/1 bufferless preemptive system, given PAoI threshold $p_i>0$, the violation probability of the PAoI corresponding to source $i$ can be expressed as
\begin{align}\label{vpdfp}
P_{i}^{\text{P}}(p_i)=1-\frac{1}{{{L_S}( \lambda  )}}\int_0^{p_i} {{f_{{Y_i}}}( y ){\text d}y\!\int_0^{p_i - y}{{e^{ - \lambda t}}{f_S}( t ){\text d}t} }.
\end{align}
The corresponding p.d.f. of the PAoI is given by
\begin{align}\label{pdfp}
{f_{\Delta _i^{\text{P}}}}(x)\!=\!
\begin{cases}
\!{\frac{1}{{{L_S}( \lambda  )}}\!\int_0^x\! {{e^{ - \lambda ( {x - y} )}}{f_{{Y_i}}}( y ){f_S}( {x - y} ){\text d}y}},&\!\!{\!{x \!\geq\!  0,}\!}\!\!\\
\!0, & \!{\!{\!{{x \!<\! 0.}}\!}\!}\!
\end{cases}
\end{align}
\end{thm}
\begin{proof}
See Appendix A.
\end{proof}

For a specific service process, $f_S(x)$ is given and ${{L_S}\left( \lambda  \right)}$ can be determined.
Note that \cite{8406928}
\begin{align}\label{ly}
{L_{{Y_i}}}\left( s \right) = {{{\lambda _i}{L_S}\left( {\lambda  + s} \right)}}/({{{\lambda _i}{L_S}\left( {\lambda  + s} \right) + s}}).
\end{align}
$f_{Y_{i}}(x)$ can be derived via inverse LT. Thus, $P_{i}^{\text{P}}(p_i)$  and ${f_{\Delta _i^{\text{P}}}}(x)$ can be derived in closed form through Theorem \ref{thm pdfp}.
\subsection{ Violation Probability of AoI in the M/G/1/1 System}
\textcolor{black}{Next, the violation probability and p.d.f. of AoI in the multi-source M/G/1/1 bufferless preemptive system are derived with a time-domain approach.}

Let us first introduce the violation time to facilitate the analysis.
\textcolor{black}{Formally, we define the time interval that the AoI of the $k$-th successfully received update of source $i$ is larger than the given AoI threshold $w_i$ during the inter-departure time $Y_{i,k+1}$, as the violation time ${J_{i,k}}$.}
Thus, according to the values of $T_{i,k}$ and $Y_{i,k+1}$, one can obtain
\begin{align}\label{p17}
{J_{i,{k}}} =
\begin{cases}
{Y_{i,k+1}},&{{T_{i,k}} > w_i,}\\
{Y_{i,k+1}}\! +\! {T_{i,k}}\! - \!w_i,&{w_i \!-\! {Y_{i,k+1}} < {T_{i,k}} \leq w_i,}\\
0,&{{T_{i,k }} \leq w_i \!-\! {Y_{i,k+1}}.}
\end{cases}
\end{align}
For instance, in terms of the age evolution shown in Fig. \ref{fig:model}, it has that ${J_{1,1}}={T_{1,1}} + {Y_{1,2}} - {w_i}$, ${J_{1,2}}={T_{1,2}} + {Y_{1,3}} - {w_i}$, ${J_{1,3}}={Y_{1,4}}$ and ${J_{1,4}}=0$.

Then, let us characterize the violation probability of AoI with ${J_{i,k}}$.
Recalling that $n_{i}(t)=\max\{k|t'_{i,k}\leq t\}$, one can get that during time $(0,t)$, the number of successfully received updates from source $i$ is $n_i(t)$.
Therefore,
the violation probability of AoI is given by
\begin{align}\label{p21}
P_{i}^{\text{A}}(w_i)=\Pr\{\Delta_{i}>w_i\}
& = \mathop {\lim }\limits_{t \to \infty } \frac{1}{{t}}\sum\nolimits_{k = 1}^{{n_i}\left( {t} \right)} {{J_{i,k}}}  \nonumber\\
& = \mathop {\lim }\limits_{t \to \infty } \frac{{{n_i}\left( {t} \right)}}{{t}}\frac{\sum\nolimits_{k = 1}^{{n_i}\left( t \right)} {{J_{i,k}}}}{{{n_i}\left( {t} \right)}}\nonumber\\
&= {\frac{1}{{\mathop {\lim }\limits_{t \to \infty } \frac{{t}}{{{n_i}\left( {t} \right)}}}}}{\mathop {\lim }\limits_{{t} \to \infty } \frac{{\sum\nolimits_{k = 1}^{{n_i}\left( {t} \right)} {{J_{i,k}}} }}{{{n_i}\left( {t} \right)}}}\nonumber\\
&= \frac{\mathbb{E}\left[ {{J_{i,k}}} \right]}{{\mathbb{E}\left[ {{Y_i}} \right]}},
\end{align}
where $\mathbb{E}[\cdot]$ denotes the expectation operator. By deriving ${\mathbb{E}\left[ {{J_{i,k}}} \right]}$ and ${{\mathbb{E}\left[ {{Y_i}} \right]}}$, the violation probability of AoI can be obtained, as shown in the following theorem.

\begin{thm}\label{cor p2}
In the multi-source M/G/1/1 bufferless preemptive system, given AoI threshold $w_i>0$, the violation probability of the AoI corresponding to source $i$ is given by
\begin{align}
P_i^{\text{A}}(w_i) &= 1 - \int_0^{w_i} {{f_{{\Delta _i}}}( x ){\text d}x},
\end{align}
where ${f_{{\Delta _i}}}( x )$ represents the p.d.f. of the AoI:
\begin{align}\label{fgfgfg}
{f_{{\Delta _i}}}\left( x \right)= L^{-1}\left[\frac{{{\lambda _i}{L_S}\left( {\lambda  + s} \right)}}{{{\lambda _i}{L_S}\left( {\lambda  + s} \right) + s}}\right].
\end{align}
\end{thm}
\begin{proof}
See Appendix B.
\end{proof}

For a specific service process, $P_{i}^{\text{A}}(w_i)$ and ${f_{\Delta _i^{\text{A}}}}(x)$ can be expressed explicitly based on Theorem \ref{cor p2} and ${{L_S}\left( \lambda  \right)}$.
\begin{rem}\label{rm1}
Note that, (\ref{ly}) and (\ref{fgfgfg}) imply
${f_{{\Delta _i}}}\left( x \right)={f_{{ Y _i}}}\left( x \right)$.
That is, in the multi-source M/G/1/1 bufferless preemptive system, the p.d.f. of the AoI of updates from source $i$ is the same as the p.d.f. of inter-departure time of the updates.  This is a unique and useful property, which indicates that for a multi-source bufferless preemptive system with Poisson arrival, no matter what the service process is,  the distribution of the corresponding AoI can be obtained via analysing the inter-departure time at the monitor.
\end{rem}
\section{Violation Probabilities in the M/M/1/1 System}\label{sec 4}
In this section, the violation probabilities and distributions of PAoI and AoI are derived in closed form, for the system with negative-exponentially distributed service time, i.e. multi-source M/M/1/1 bufferless preemptive system.

Based on Theorem \ref{thm pdfp} and Theorem \ref{cor p2}, it can be found that to derive the violation probabilities of AoI and PAoI in the M/M/1/1 system, one needs to know the p.d.f.s of $S$ and $Y_i$.

First, the p.d.f. of $S$ can be formally expressed as
\begin{align}\label{p29}
{f_S}\left( x \right) =
\begin{cases}
{\mu {e^{ - \mu x}}},&{x \geq 0,}\\
0,&{x < 0.}
\end{cases}
\end{align}
The LT of ${f_S}\left( x \right)$ can be directly obtained:
\begin{align}\label{p30}
{L_S}\left( s \right) = {\mu }/({{s + \mu }}).
\end{align}
Based on (\ref{p29}) and (\ref{p30}),  $f_{Y_{i}}(x)$ is given in the following.
\begin{lem}\label{lem 1}
In the multi-source M/M/1/1 bufferless preemptive system, the p.d.f. of the inter-departure time corresponding to source $i$ is given by
\begin{align}\label{p31}
{f_{Y_i}}(x)=
\begin{cases}
{\frac{{{\lambda _i}\mu }}{{a_i - b_i}}\left( {{e^{a_ix}} - {e^{b_ix}}} \right)},&{x \geq 0,}\\
0,&{x < 0,}
\end{cases}
\end{align}
where $a_i$ and $b_i$ are the solutions of the quadratic equation w.r.t. $s$, ${{s^2} + \left( {\lambda  + \mu } \right)s + {\lambda _i}\mu }=0$.
\end{lem}
\begin{proof}
See Appendix C.
\end{proof}
Note that
\begin{align}\label{p37}
a_i+b_i=-(\lambda+\mu) \;{\text{and}}\;a_ib_i=\lambda_i \mu.
\end{align}
With these results, the violation probabilities and distributions of PAoI and AoI can be derived explicitly.
\subsection{ Violation Probability  of PAoI in the M/M/1/1 System}
Based on Theorem \ref{thm pdfp}, the violation probability and p.d.f. of PAoI are given by the following corollary.
\begin{cor}\label{cor p3}
In the multi-source M/M/1/1 bufferless preemptive system, given PAoI threshold $p_i>0$, the violation probability of PAoI corresponding to source $i$ is given by
\begin{align}
P_i^{\text P}(p_i) = {e^{ - \left( {\lambda  + \mu } \right)p_i}} + \frac{{\lambda  + \mu }}{{a_i - b_i}}\left( {{e^{a_ip_i}} - {e^{b_ip_i}}} \right).
\end{align}
The corresponding p.d.f. of PAoI can be expressed as
\begin{align}
{f_{\Delta _i^{\text{P}}}}(x) \!=\!
\begin{cases}
{\left( {\lambda \! +\! \mu } \right)\!( {{e^{ - \left( {\lambda  + \mu } \right)x}} \!+ \!\frac{{b_i{e^{b_ix}} - a_i{e^{a_ix}}}}{{a_i - b_i}}} )},&{x \!\geq \!0,}\\
0,&{x\! <\! 0.}
\end{cases}
\end{align}
\end{cor}
\begin{proof}
See Appendix D.
\end{proof}

\begin{rem}
From Corollary \ref{cor p3}, the average PAoI is given by
\begin{align}\label{apaoi}
\mathbb{E}\!\left[ {\Delta _i^{\text P}} \right]\!= \!\int_0^{ + \infty } \!{x{f_{\Delta _i^{\text{P}}}}(x){\text d}x} \mathop {\rm{ = }}\limits^{\left( {\text a} \right)}  \frac{1}{{\left( {\lambda  + \mu } \right)}} + \frac{{\lambda  + \mu }}{{{\lambda _i}\mu }},
\end{align}
in which the equality (a) holds following from (\ref{p37}).
Moreover, the mean square and variance of PAoI are given by
\begin{align}
\mathbb{E}\!\left[ {\!{{\left( {\Delta _i^{\text P}} \right)}^2}\!} \right]\!
= \!\!\int_0^{ + \infty }\! \!\!{{x^2}{f_{\Delta _i^{\text{P}}}}(x){\text d}x}
  \mathop {\rm{ = }}\limits^{\left( {\text b} \right)} \!\frac{2}{{{{\left( {\lambda \! +\! \mu } \right)}^2}}} \!+\! \frac{{2{{\left( {\lambda \! +\! \mu } \right)}^2}}}{{{\lambda _i}^2{\mu ^2}}},\\
\!\!\!\!\!\mathbb{D}\!\left[ {\Delta _i^{\text P}} \right]\!\! = \!\mathbb{E}\!\left[ {\!{{\left( {\!\Delta _i^{\text P}\!} \right)}^2}\!} \right]\!\! - \!{{\mathbb E}^2}\!\!\left[ {\Delta _i^{\text P}} \right]\! = \! \frac{1}{{{{\left( {\lambda \! + \!\mu } \right)}^2}}}\! + \!\frac{{{{\left( {\lambda\!  +\! \mu } \right)}^2}}}{{{\lambda _i}^2{\mu ^2}}}\!-\! \frac{2}{{{\lambda _i}\mu }}\!.\label{vp}
\end{align}
in which the equality (b) holds  following from (\ref{p37}).
\end{rem}
\subsection{Violation Probability of AoI in the M/M/1/1 System}
Based on Theorem \ref{cor p2}, the violation probability and p.d.f. of AoI are given by the following corollary.
\begin{cor}\label{cor p5}
In the multi-source M/M/1/1 bufferless preemptive system, given AoI threshold $w_i>0$, the violation probability of AoI corresponding to source $i$ is given by
\begin{align}
P_i^{\text A}(w_i) = \frac{{a_i{e^{b_iw_i}} - b_i{e^{a_iw_i}}}}{{a_i - b_i}}.
\end{align}
The corresponding p.d.f. of AoI can be expressed as
\begin{align}\label{yui}
{f_{{\Delta _i}}}\left( x \right)=
\begin{cases}
{\frac{{{\lambda _i}\mu }}{{a_i - b_i}}\left( {{e^{a_ix}} - {e^{b_ix}}} \right)},&{x \geq 0,}\\
0,&{x < 0.}
\end{cases}
\end{align}
\end{cor}
\begin{proof}
Based on Remark \ref{rm1} and Lemma \ref{lem 1}, the p.d.f. of AoI is directly obtained, as shown by (\ref{yui}). Accordingly, the violation probability of AoI is derived as follows,
\begin{align}
 P_i^{\text A}\!(w_i) \!=\! \Pr\! \left\{ {\!{\Delta _i} > w_i\!} \right\} \!= \!1\! -\!\! \int_0^{w_i} \!\!\!{{f_{{\Delta _i}}}\!(x){\text d}x} \!=\! \frac{{a_i{e^{b_iw_i}}\! - \! b_i{e^{a_iw_i}}}}{{a_i - b_i}}\!.\nonumber
\end{align}
This ends the proof.
\end{proof}

\begin{rem}
Based on Corollary \ref{cor p5}, the average AoI is given by
\begin{align}\label{aaoi}
\mathbb{E}\left[ {{\Delta _i}} \right] = \int_0^{ + \infty } {x{f_{{\Delta _i}}}(x){\text d}x}    \mathop {\rm{ = }}\limits^{\left( {\text a} \right)}   \frac{{\lambda  + \mu }}{{{\lambda _i}\mu }},
\end{align}
in which the equality (a) holds following from (\ref{p37}).
The mean square and variance of AoI are given by
\begin{align}
\mathbb{E}\!\left[ {{{\left( {{\Delta _i}} \right)}^2}} \right]= \int_0^{ + \infty }\!\!\! {{x^2}{f_{{\Delta _i}}}(x){\text d}x} \mathop {\rm{ = }}\limits^{\left( {\text b} \right)}   \frac{{2{{\left( {\lambda  + \mu } \right)}^2}}}{{{\lambda _i}^2{\mu ^2}}} - \frac{2}{{{\lambda _i}\mu }},\\
\mathbb{D}\!\left[ {\Delta _i} \right] = \mathbb{E}\!\left[ {{{\left( {\Delta _i} \right)}^2}} \right] - {{\mathbb E}^2}\!\left[ {\Delta _i} \right] = \frac{{{{\left( {\lambda  + \mu } \right)}^2}}}{{{\lambda _i}^2{\mu ^2}}} - \frac{2}{{{\lambda _i}\mu }}.\label{va}
\end{align}
in which the equality (b) holds  following from (\ref{p37}).
\end{rem}

\textcolor{black}{(\ref{apaoi}) and (\ref{aaoi}) are consistent with the results of \cite{8406928}.}

\begin{rem}\label{ioi}
The variances of AoI and PAoI characterize the stability of the AoI process. According to (\ref{vp}) and (\ref{va}), it is found that the first derivatives of both $\mathbb{D}\!\left[ {\Delta _i} \right]$ and $\mathbb{D}\!\left[ {\Delta _i^{\text P}} \right]$ are ${{2( {{\lambda _i}\mu  - {{( {\lambda  + \mu } )}^2}} )} /({\lambda _i^3{\mu ^2}})}<0$. This indicates that both $\mathbb{D}\!\left[ {\Delta _i} \right]$ and $\mathbb{D}\!\left[ {\Delta _i^{\text P}} \right]$ monotonously decrease with $\lambda_i$. That is, the stability of the AoI process of source $i$ can be improved with an increased $\lambda_i$.
\end{rem}

\section{Optimizing Overall Timeliness}\label{sec 5}
Recall that in the considered system, the total arrival rate is constrained due to the resource limitation of the IoT device.
In this section, we investigate how to allocate the arrival rates for sources to achieve the optimal overall timeliness under the resource constraint of IoT device, in terms of violation probabilities of PAoI and AoI respectively. 

To obtain more insights, we focus on the multi-source status update system with negative-exponentially distributed service time.
\textcolor{black}{The reason is as follows. Usually in IoT-based status update systems, the transmission rate, packet length and one-time transmission time are fixed.
Meanwhile, the service time, i.e., total transmission time is usually varying, since that the packet loss occasionally happens due to fluctuations of wireless channel and that retransmission  mechanism is often adopted to ensure the reliability of transmission.
To match the service time with the considered continuous AoI process, it is common to consider the service time as continuously varying and negative-exponentially distributed \cite{9217386,8469047,9705518,9732416,9312180,9324753,9103131}, since the negative-exponential distribution can be regarded as a continuous simulation of geometric distribution.}

In practice, different status update streams have different requirements on the probability of age threshold exceeding events. For instance, in the autonomous driving system, the different tolerance of outdated speed and engine temperature information would lead to different levels of the risk for driving. Low PAoI and AoI thresholds should be preset for status updates of sources holdong the potential of high risk. Therefore, it is common that $p_i\neq p_j$ and $w_i\neq w_j$ for $i\neq j$.
Motivated by this, the maximal violation probabilities
\begin{align}
P_{\text M}^{\text P}(\boldsymbol {\lambda};\boldsymbol {p}) &:= \max_{i\in \mathcal {N}}\{ P_i^{\text{P}}(\lambda_i;p_i)\}, \\
P_{\text M}^{\text A} (\boldsymbol {\lambda};\boldsymbol {w})&:= \max_{i \in  \mathcal {N}}\{ P_i^{\text{A}}(\lambda_i;w_i)\},
\end{align}
are introduced to measure the overall timeliness. $P_{\text M}^{\text P}(\boldsymbol {\lambda};\boldsymbol {p})$ (or $P_{\text M}^{\text A}(\boldsymbol {\lambda};\boldsymbol {w})$) is the maximum among the violation probabilities of PAoI (or AoI) of all sources, where $\boldsymbol {\lambda}=(\lambda_1,\lambda_2,\ldots,\lambda_N )$, $\boldsymbol {p}=(p_1,p_2,\ldots,p_N )$ and $\boldsymbol {w}=(w_1,w_2,\ldots,w_N )$.

Given total arrival rate $\lambda$, let us investigate the arrival rate allocation $\boldsymbol {\lambda}$  minimizing the maximal violation probabilities $P_{\text{M}}^{\text{P}}$ or $P_{\text{M}}^{\text{A}}$. Accordingly, the corresponding two optimization problems are formulated as follows:
\textcolor{black}{
\begin{align}
\text{Problem }{\mathcal{P}_1 :}\label{P1}
\mathop {\min }\limits_{\boldsymbol {\lambda} }  &\quad P_{\text M}^{\text P}(\boldsymbol {\lambda};\boldsymbol {p})= \max_{i\in \mathcal {N}}\{ P_i^{\text{P}}(\lambda_i;p_i)\} \\
 {\text{s}}.\;{\text{t}}.& \quad\sum\nolimits_{i = 1}^N  {{\lambda _i} - \lambda  = 0}, \tag{\ref{P1}{a}}\label{P1a}\\
 & \quad{\lambda _i} > 0,\forall i \in \mathcal{N} ;\tag{\ref{P1}{b}}\\
\text{Problem }{\mathcal{P}_2 :}\label{P2}
\mathop {\min }\limits_{\boldsymbol {\lambda} }  &\quad P_{\text M}^{\text A} (\boldsymbol {\lambda};\boldsymbol {w})= \max_{i \in  \mathcal {N}}\{ P_i^{\text{A}}(\lambda_i;w_i)\}\\
 {\text{s}}.\;{\text{t}}.&\quad \sum\nolimits_{i = 1}^N  {{\lambda _i} - \lambda  = 0}, \tag{\ref{P2}{a}}\label{P2a}\\
\;\quad \quad \quad \ &\quad {\lambda _i} > 0,\forall i \in \mathcal{N} .\tag{\ref{P2}{b}}
\end{align}}
To solve ${\mathcal{P}_1}$ and ${\mathcal{P}_2}$, several findings on the properties of both problems are presented by the following propositions.
\begin{prop}\label{pp1}
$\mathcal{P}_1$ is a convex optimization problem. Additionally, $P_i^{\text{P}}(\lambda_i;p_i)$ monotonically decreases with $\lambda_i$, regardless of $\lambda_j$, for all $j\neq i$.
\end{prop}
\begin{proof}
See Appendix E.
\end{proof}

\begin{prop}\label{pp2}
$\mathcal{P}_2$ is a convex optimization problem. Additionally, $P_i^{\text{A}}(\lambda_i;w_i)$ monotonically decreases with $\lambda_i$, regardless of  $\lambda_j$, for all $j\neq i$.
\end{prop}
\begin{proof}
See Appendix F.
\end{proof}

As Propositions \ref{pp1} and \ref{pp2} reveal that both $\mathcal{P}_1$ and $\mathcal{P}_2$ are convex, the optimal arrival rate allocations can be obtained by employing standard convex optimization algorithms. Note that there are many low complexity algorithms for convex optimization problems, e.g., Newton method with log-barrier \cite{23081055}. We will further investigate the performance achieved by convex optimization with numerical results.
\textcolor{black}{In addition, a sufficient and necessary condition for achieving the minimum of the maximal violation probability of AoI (or PAoI) is shown in the following proposition.}
\textcolor{black}{
\begin{prop}\label{pp3}
The minimum of the maximal violation probability of AoI (or PAoI) is achieved only if the violation probabilities of AoI (or PAoI) for all sources are equal.
\end{prop}
\begin{proof}
Let us take $\mathcal{P}_1$ as an example, as $\mathcal{P}_1$ and $\mathcal{P}_2$ possess similar properties. Recall that $P_i^{\text{A}}$ monotonically decreases with $\lambda_i$, regardless  of $\lambda_j$, for all $j\neq i$.
\textcolor{black}{Denote the optimal arrival rate allocation and the arrival rate allocation when the violation probabilities of AoI are equal, by $\boldsymbol {\lambda}^*$ and $\boldsymbol {\lambda}^{\dag}$, respectively.
Suppose there exists a set $\mathcal {M}  \subset  \mathcal {N}$, such that for all $m\in \mathcal {M}$, $\lambda_m^*\neq \lambda_m^{\dag}$. Note that $  \sum\nolimits_{m \in \mathcal{M}} {\lambda _i^*}= \sum\nolimits_{m \in \mathcal{M}} {\lambda _i^\dag} $.
Hence, it is deduced that there must exist $m'\in \mathcal{M}$ such that $\lambda_{m'}^*<\lambda_{m'}^\dag$, by using the pigeonhole principle.
This implies $P_{\text{M}}^{{\text{A}}}(\boldsymbol {\lambda}^*)\geq P_{m'}^{{\text{A}}}(\lambda_{m'}^*) > P_{m'}^{{\text{A}}}(\lambda_{m'}^\dag) =P_{\text{M}}^{{\text{A}}}(\boldsymbol {\lambda}^\dag)$, following from the monotonicity of $P_i^{\text{A}}(\lambda_i)$.
Note that $P_{\text{M}}^{{\text{A}}}(\boldsymbol {\lambda}^*)>P_{\text{M}}^{{\text{A}}}(\boldsymbol {\lambda}^\dag)$ contradicts with the fact $P_{\text{M}}^{{\text{A}}}(\boldsymbol {\lambda}^*)$ is the minimum of the maximal violation probability of AoI.
Hence, one can conclude that $\mathcal {M}$ should be an empty set, which implies $\boldsymbol {\lambda} ^*=\boldsymbol {\lambda}^\dag$.}
\end{proof}}

\section{Numerical Results}\label{sec num}
In this section, we first present the violation probabilities of AoI and PAoI as functions of AoI and PAoI thresholds, respectively.
The simulations are also provided as validation checks, where $6\times 10^5$ updates are generated and $1 \times 10^6$ moments are sampled.
Moreover, the effects of service distribution on the violation probabilities, and that of arrival rate allocation on the maximal violation probabilities are analyzed.
Finally, we study the improvement of the overall timeliness achieved by the optimal arrival rate allocations.
\subsection{Violation Probabilities of AoI and PAoI}
Let us show the violation probabilities of AoI and PAoI for the M/M/1/1 system by Fig. \ref{fig:vp}.
Source $1$ is focused on.
\begin{figure}[!t]
  \subfigure[Violation probability of AoI.]{
    \label{fig:vpaoi} 
    \includegraphics[width=0.238\textwidth]{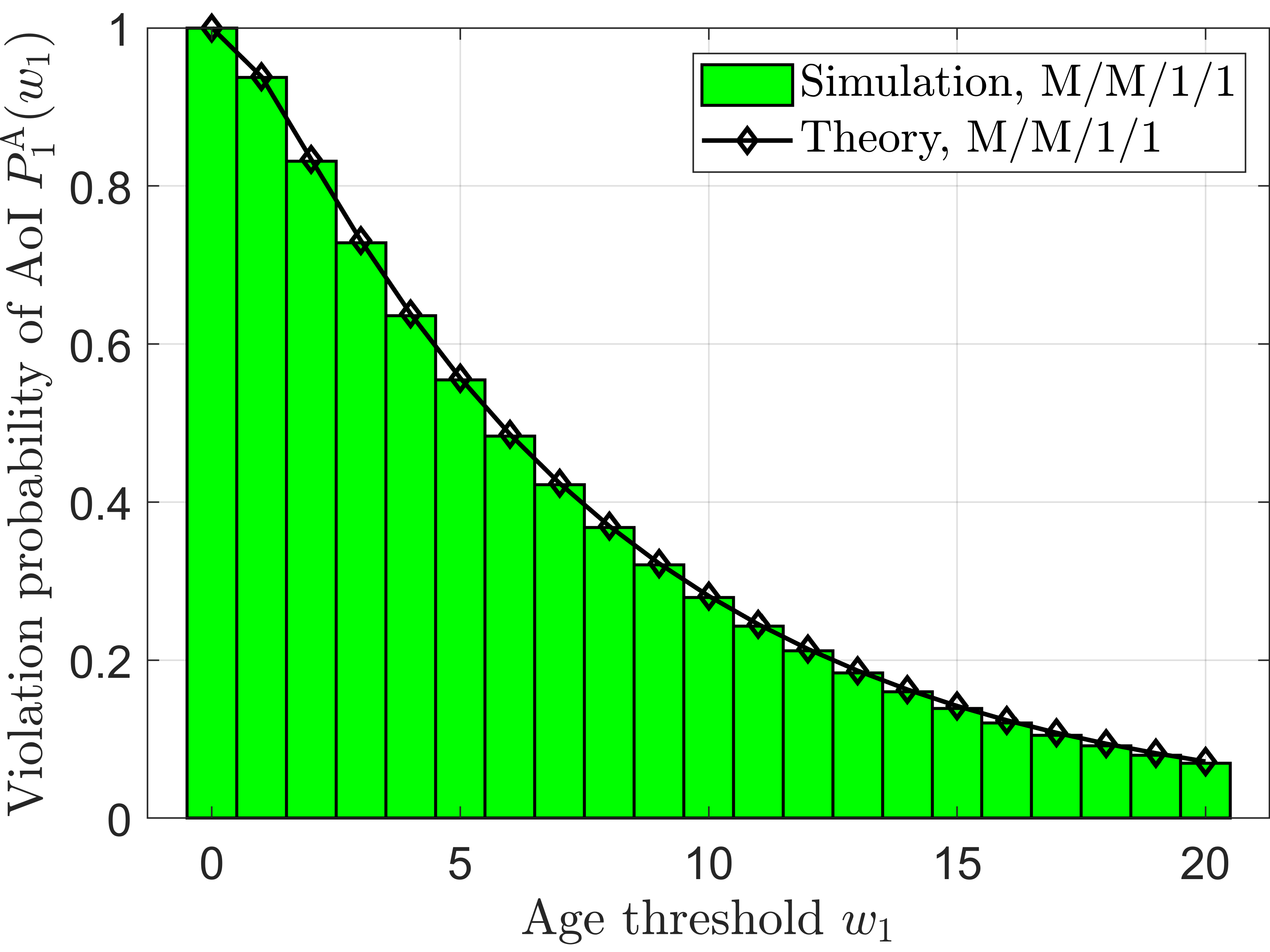}}
  \subfigure[Violation probability of PAoI.]{
    \label{fig:vppaoi} 
    \includegraphics[width=0.238\textwidth]{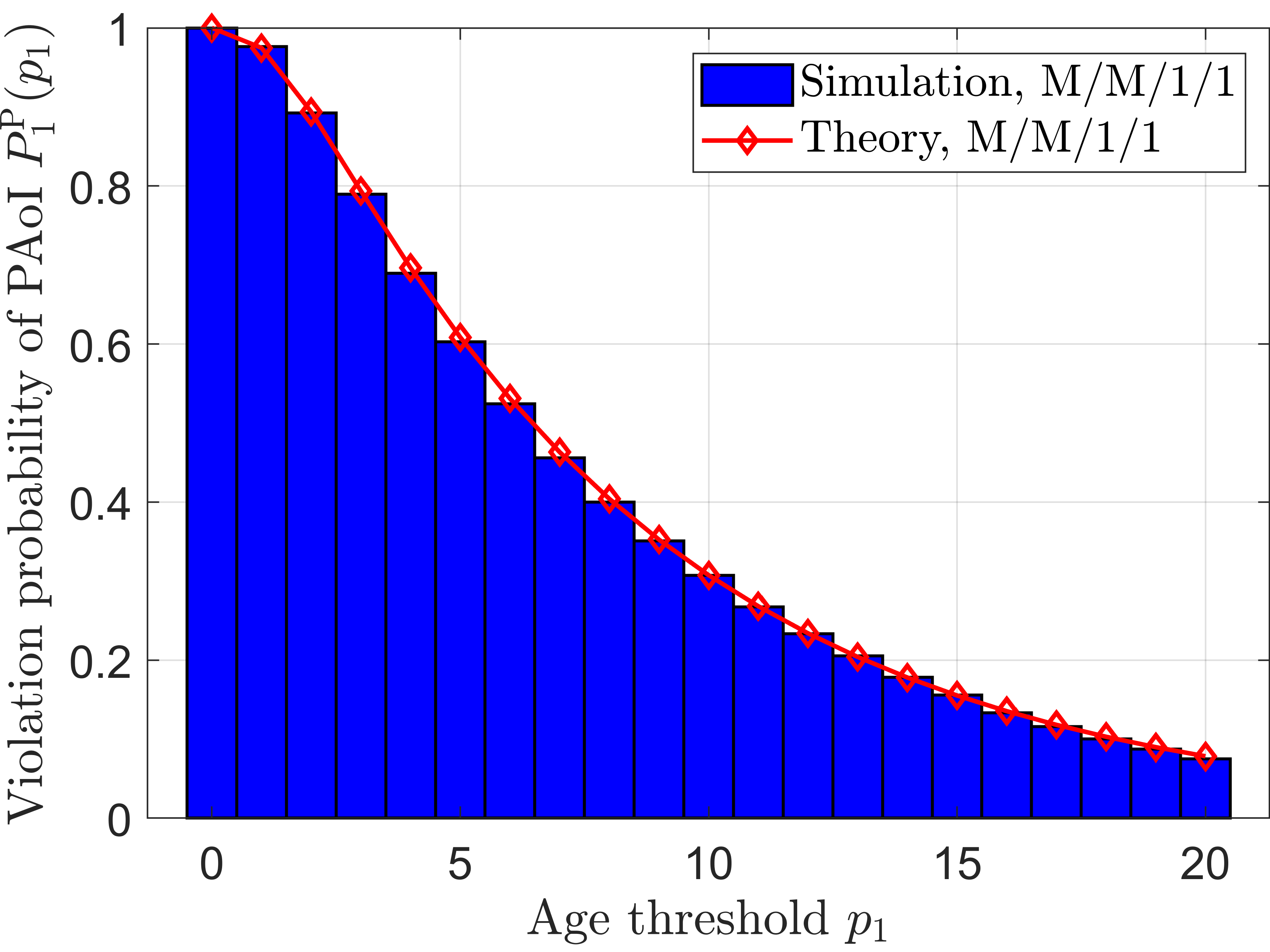}}
  \caption{Violation probabilities of AoI and PAoI of source $1$ as functions of the corresponding thresholds. $\mu=1$, $\lambda=0.6$ and $\lambda_1=0.2$.}
  \label{fig:vp} 
\end{figure}
Fig. \ref{fig:vp}\subref{fig:vpaoi} presents a black curve and green bars, which reflect the theoretical and simulation values of $P_1^{\text A}(w_1)$, respectively.
First, it is found that the theoretical result matches well with the simulation result, which indicates the correctness of Corollary \ref{cor p2}.
Moreover, it can be seen that $P_1^{\text A}(w_1)$ monotonically  decreases with $w_1$, which is consistent with the intuition.
Besides, with the same system parameters, the violation probability of PAoI is shown in Fig. \ref{fig:vp}\subref{fig:vppaoi}.
The red curve and the blue bars respectively represent the theoretical and simulation values of violation probability of PAoI of source $1$.
Since they are extremely close, the validation of Corollary \ref{cor p3} is confirmed.
\subsection{Effect of Service Time Distribution}
\begin{figure}[!t]
\centering
\includegraphics[width=0.75\linewidth]{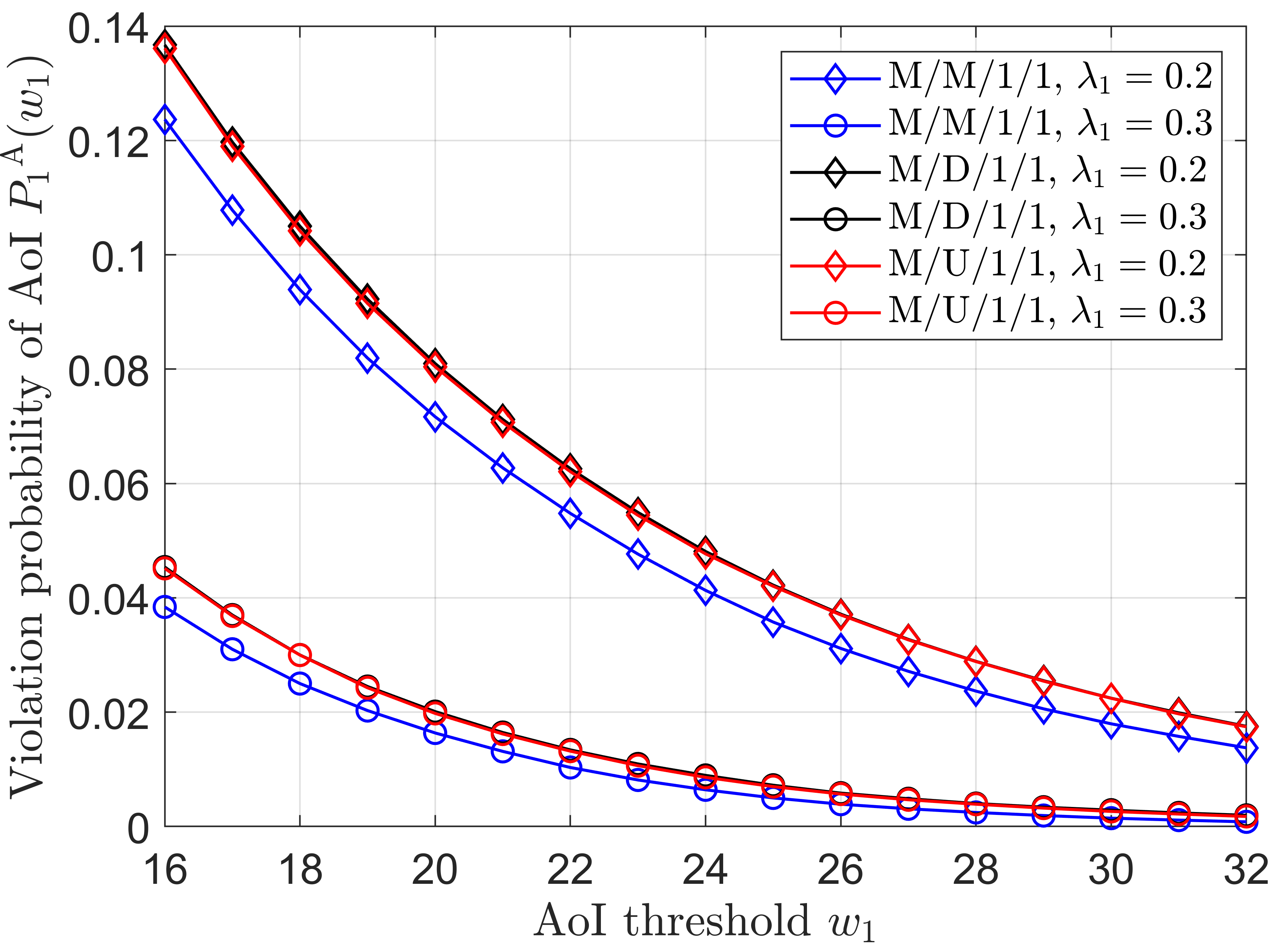}
\caption{Violation probabilities of AoI of source $1$ in the M/M/1/1, M/D/1/1 and M/U/1/1 systems as functions of the AoI threshold $w_1$. $\mu=1$ and $\lambda=0.6$.}
\label{fig:cdf}
\end{figure}
The effect of service time distribution is studied, as shown in Fig. \ref{fig:cdf}.
Specifically, the negative-exponentially distributed, deterministic and uniformly distributed service times are considered as examples.
We focus on Source 1.
The blue, black and red curves represent the violation probabilities of AoIs with negative-exponentially distributed, deterministic, and uniformly distributed service times, respectively.
Curves with diamond or circle markers correspond to the cases where $\lambda_1$ $=$ $0.2 \text{ or } 0.3$, respectively.
First, it is shown that the violation probability with negative-exponentially distributed service time is lower than that with uniformly distributed service time or deterministic distributed service time.
This indicates that the M/M/1/1 system outperforms the M/U/1/1 and M/D/1/1 systems.
\textcolor{black}{The reason is that the negative-exponentially distributed service time possesses the memoryless property, which is a significant property to improve the timeliness.
Specifically, apart from that the preemption avoids waiting, when the preemption occurs, the memoryless service time enables the preemptive queueing system to serve the new update within the same time as that used for continuing serving the replaced update if there is no preemption. This further lessens the system time.}
\textcolor{black}{
Furthermore, it is found that the violation probabilities of the M/D/1/1 and M/U/1/1 systems are quite close.
This is because the variances of service times of  M/D/1/1 and M/U/1/1 systems are low.
Hence, one can conclude that, when given the average service time,
the effect of service time distribution is more significant for the system with
greater variance of service time.}
One can also find that for the same service time distribution, the curve with circle markers is below that with diamond markers. This indicates that, given the total arrival rate, the lower violation probability of a source can be achieved by increasing the corresponding arrival rate, which validates Proposition \ref{pp2}. This is consistent with the fact that if a larger arrival rate is allocated to a source, its timeliness will be better since its updates have a lower probability of being preempted by other sources.

\subsection{Effect of Arrival Rate Allocation}
Next, let us analyse the effect of the arrival rate allocations on the overall timeliness based on the maximal violation probabilities of AoI and PAoI, as shown in Fig. \ref{fig:mvp}. Without loss of generality, a dual-source M/M/1/1 system is considered. By noticing the similarities between  Fig. \ref{fig:mvp}\subref{fig:mvpaoi} and \ref{fig:mvp}\subref{fig:mvppaoi}, we thus take Fig. \ref{fig:mvp}\subref{fig:mvpaoi} as an example and analyse how the arrival rate allocation affects the maximal violation probability of AoI. The obtained results are also applicable for the PAoI case.

\begin{figure}[!t]
  \centering
  \subfigure[$w_2=10$ and $w_1=10, 5, 15$ respectively.]{
    \label{fig:mvpaoi} 
    \includegraphics[width=0.35\textwidth]{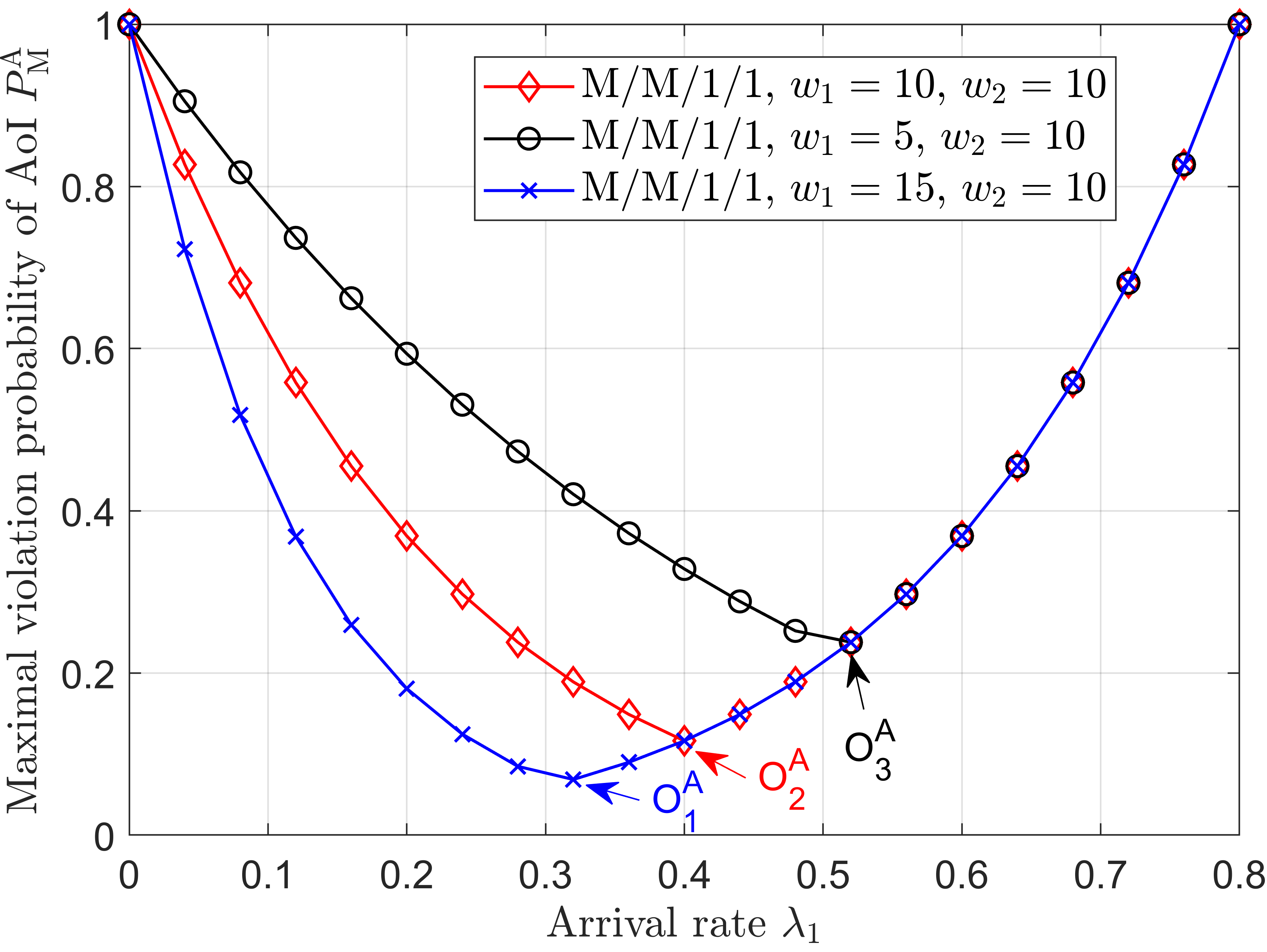}}\hfill\\
  \subfigure[$p_2=10$ and $p_1=10, 5, 15$ respectively.]{
    \label{fig:mvppaoi} 
    \includegraphics[width=0.35\textwidth]{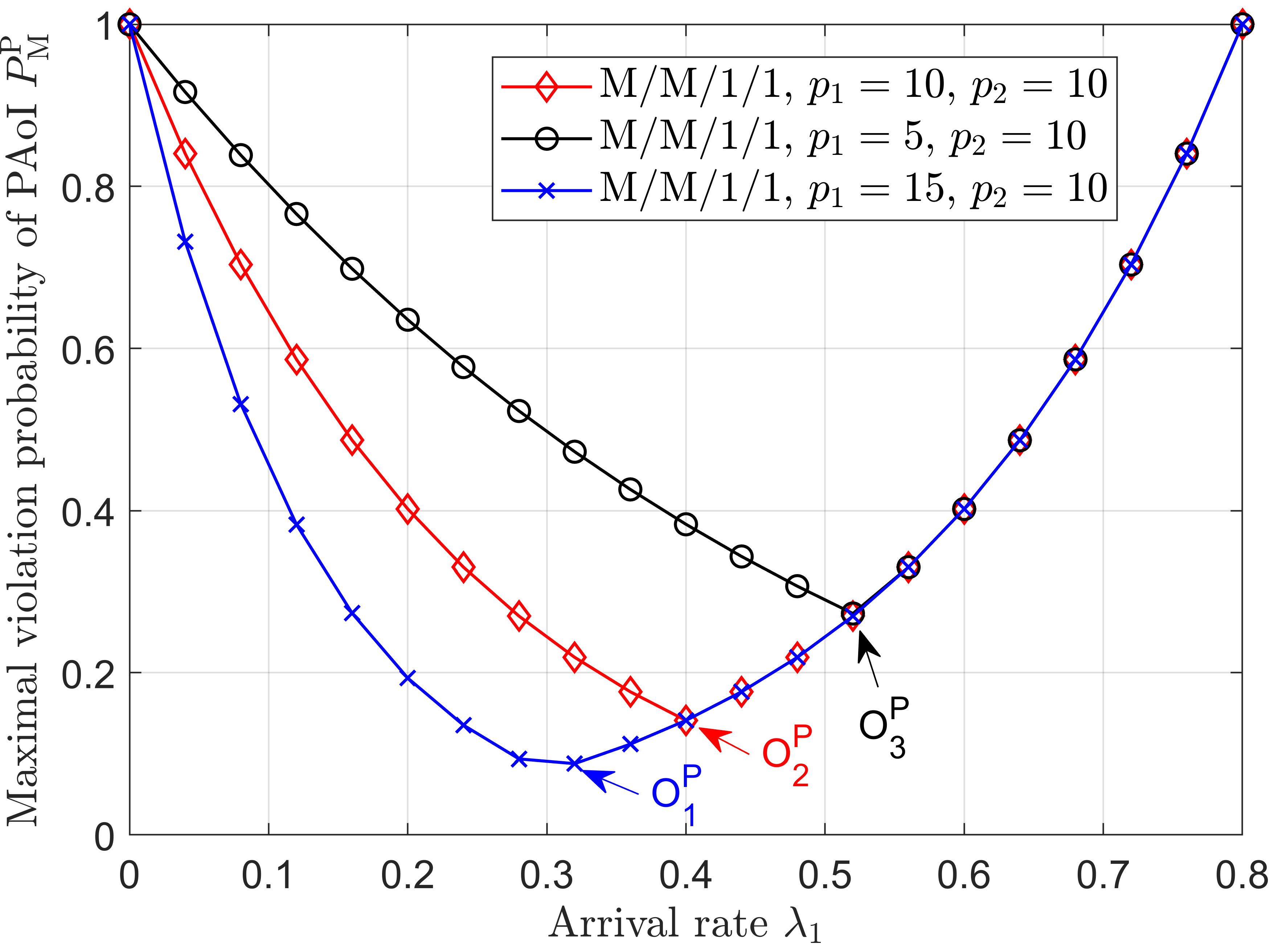}}
  \caption{Maximal violation probabilities of AoI and PAoI as functions of the arrival rate of source 1. $\mu=1$ and $\lambda=0.8$.}
  \label{fig:mvp} 
\end{figure}

Fig. \ref{fig:mvp}\subref{fig:mvpaoi} depicts how the maximal violation probabilities of AoI vary with the arrival rate of source $1$, where the blue, red and black curves correspond to the case when setting $w_2=10$ and $w_1= 15, 10, 5$, respectively.
In addition, there are points ${\text O_1^\text A}$, ${\text O_2^\text A}$ and ${\text O_3^\text A}$ representing the optimal arrival rate allocations and the minimums of maximal violation probabilities of AoI for corresponding AoI thresholds. These points are found by a convex optimization algorithm utilizing the Newton method with log-barrier, where the initial arrival rate allocation is $\left( {{\lambda _1},{\lambda _2}} \right) = \left( {0.4,0.4} \right)$ and the step size is chosen by the backtracking line search.
Since $P_i^{\text A}$ monotonically decreases with $\lambda_i$, it is also important to note that, for each curve, the part to the left of the optimal point represents $P_1^{\text{A}}$ while that to the right of the optimal point represents $P_2^{\text{A}}$. Note that when the optimal points are achieved, the violation probabilities for all sources are equal, which verifies Proposition \ref{pp3}.
Moreover, it is shown that the curves are convex, which implies the correctness of Proposition \ref{pp1}. Furthermore, as shown by the optimal points, it can be obtained that the optimal arrival rate allocations reduce the maximal violation probabilities effectively, i.e., improve the overall timeliness of multi-source status update system significantly. Besides, it is found that ${\text O_3^\text A}$ is to the right of ${\text O_2^\text A}$ and ${\text O_1^\text A}$ is to the left of ${\text O_2^\text A}$. This implies that, with the decrease of the AoI threshold of source $1$, the optimal arrival rate of source $1$ increases. Therefore, to achieve the better overall timeliness, one should allocate higher arrival rates for the sources with higher timeliness requirements.

\subsection{Overall Timeliness Gains from the Optimal Arrival Rates}
Finally, Fig. \ref{fig:mmvp} is provided for studying the improvement of the maximal violation probabilities of AoI and PAoI when the corresponding optimal arrival rate allocations are adopted. Without loss of generality, the dual-source M/M/1/1 system is considered. Note that Fig. \ref{fig:mmvp}\subref{fig:ggg} is similar to Fig. \ref{fig:mmvp}\subref{fig:ttt}, hence we take Fig. \ref{fig:mmvp}\subref{fig:ttt} as an example and analyse the improvement of the maximal violation probability of AoI in the following. The obtained results are also applicable for the PAoI case.
\begin{figure}[!t]
  \centering
  \subfigure[$\left( {{w_1},{w_2}} \right) =$ (7.5, 7.5), (5, 10), (2, 13), respectively.]{
    \label{fig:ttt} 
    \includegraphics[width=0.35\textwidth]{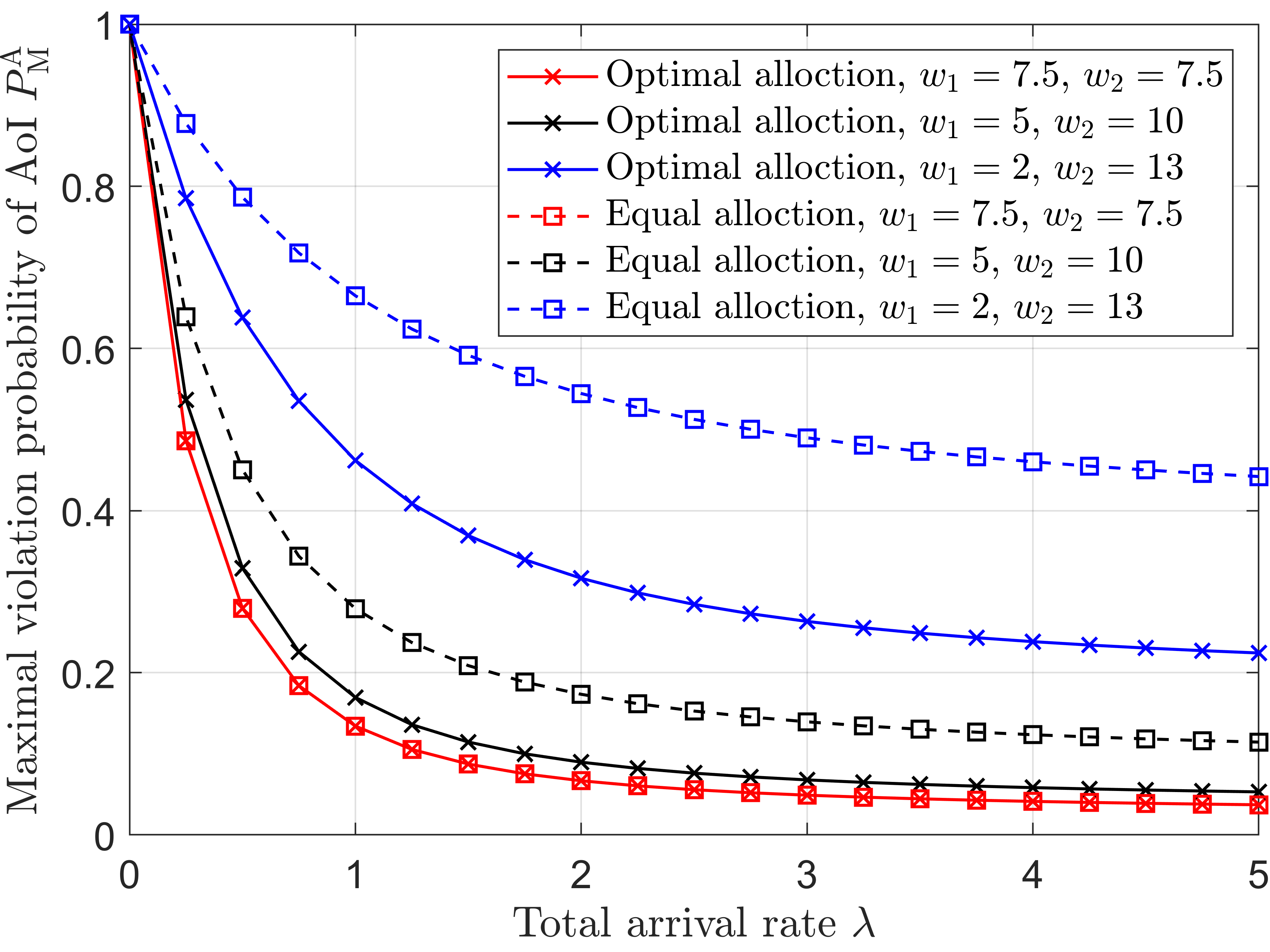}}\\
  \subfigure[$\left( {{p_1},{p_2}} \right)=$ (7.5, 7.5), (5, 10), (2, 13), respectively.]{
    \label{fig:ggg} 
    \includegraphics[width=0.35\textwidth]{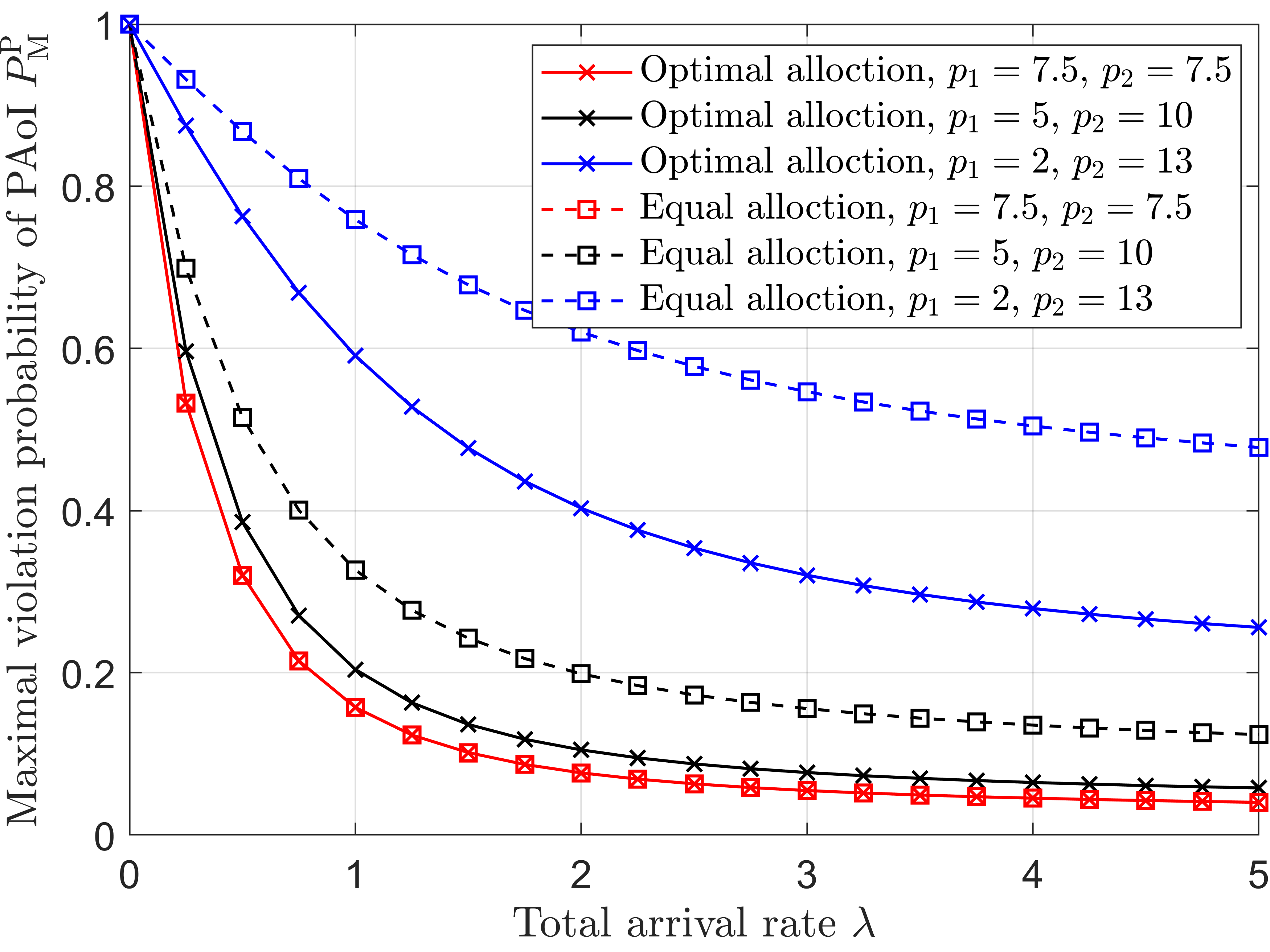}}
  \caption{Maximal violation probabilities  as functions of total arrival rate. $\mu$$=$$1$.}
  \label{fig:mmvp} 
\end{figure}

Fig. \ref{fig:mmvp}\subref{fig:ttt} compares the optimal arrival rate allocations with the commonly used equal arrival rate allocations, in which the red, black and blue curves correspond to the maximal violation probabilities of AoI when setting $( {{w_1},{w_2}} ) = ( {7.5,7.5} ),\;( {5,10} )$ and $( {2,13} )$, respectively. Also note that the cross markers and the square markers stand for optimal arrival rates and the equal arrival rate allocations, respectively.
First, it is found that the maximal violation probabilities of AoI adopting the optimal arrival rate allocation are lower than that adopting the equal arrival rate allocations. This means that the optimal arrival rate allocation performs much better than the equal arrival rate allocation. Further, one can notice that the red curves with cross markers and square markers coincide, which can be explained by the fact that the AoI thresholds for source 1 and 2 are equal and the equal arrival rate allocation is optimal in this case.
Furthermore, it is shown that the violation probability of AoI monotonously decreases with the total arrival rate, which indicates that the overall timeliness can be further improved by increasing the total arrival rate.
\section{Conclusion}\label{sec con}
We studied the violation probabilities and distributions of AoI and PAoI for the multi-source M/G/1/1 bufferless preemptive system. \textcolor{black}{The general formulas of the violation probabilities and p.d.f.s of AoI and PAoI were derived based on the time-domain analyses on the violation time, which is an insightful approach with obvious physical meaning.} Moreover, according to the obtained general formulas, we derived the corresponding violation probabilities and p.d.f.s in closed forms for the M/M/1/1 system.
\textcolor{black}{It was also proved that the per-source violation probabilities and variance monotonously decrease with the arrival rate. This indicates that the timeliness and its stability can be enhanced by increasing the corresponding arrival rate.}
Furthermore, the maximal violation probability of AoI (or PAoI), i.e., the maximum of all violation probabilities of per-source AoI (or PAoI), was proposed to characterize the overall timeliness. \textcolor{black}{To improve the overall timeliness under the resource constraint of the IoT device, we further optimized the arrival rate allocation for minimizing the maximal violation probability of AoI (or PAoI).} \textcolor{black}{We proved the convexity of the formulated problems with the deliberate mathematical analyses. The optimal arrival rate allocations were found by using standard convex optimization algorithms.} Based on the pigeonhole principle, it was also obtained that the minimum of the maximal violation probability of AoI (or PAoI) is achieved only when the violation probabilities of AoI (or PAoI) of all sources are equal. Numerical results validated the theoretical analyses and show that the optimal arrival rate allocation improves the overall timeliness significantly. \textcolor{black}{In addition, it was shown that the M/M/1/1 system outperforms M/D/1/1 and M/U/1/1 systems in terms of the violation probability of AoI.}

\appendix
\subsection{Proof of Theorem \ref{thm pdfp}}\label{prf pdfp}
To get the violation probability of PAoI, let us derive the cumulative distribution function (c.d.f.) of the PAoI of updates from source $i$, which is denoted by $F_{\Delta_{i}^{\text{P}}}(x)$.
Based on (\ref{p2}), when $x\geq0$, one can get
\begin{align}
{F_{\Delta _i^{\text{P}}}}(x) \!=\! \Pr \!\left\{ {\Delta _{i,k}^{\text{P}} \!\le\! x} \right\}
& \!=\! \Pr \! \left\{ {{Y_{i,k+1}} + {T_{i,k }} \le x} \right\}\nonumber\\
& \!=\!\! \int_0^x \!\!{{\text d}y\!\!\int_0^{x \!-\! y}\!\! \!{{f_{{Y_{i,k+1}},{T_{i,k }}}}\!\left( {y,t} \right)} } {\text d}t,\label{p10}
\end{align}
where ${{f_{{Y_{i,k+1}},{T_{i,k}}}}( {y,t} )}$ is the joint p.d.f. of inter-departure time $Y_{i,k+1}$ and system time ${T_{i,k }}$.

Since $Y_{i,k+1}$ is independent with $T_{i,k}$ and the p.d.f. of $T_{i,k}$ is given by \cite{8406928}
\begin{align}\label{p8}
f_{T}(x)= \frac{e^{-\lambda x}f_{S}(x)}{L_{S}(\lambda )},
\end{align}
 one can get
\begin{align}\label{p11}
{{f_{{Y_{i,k+1}},{T_{i,k}}}}( {y,t} )}= f_{Y_i}(y)f_{T}(t)= f_{Y_i}(y)\frac{e^{-\lambda t}f_{S}(t)}{L_{S}(\lambda )}.
\end{align}
Plunging (\ref{p8}) and (\ref{p11}) into (\ref{p10}), it has that
\begin{align}
{F_{\Delta _i^{\text{P}}}}(x) = \frac{1}{{{L_S}( \lambda  )}}\int_0^x {{f_{{Y_i}}}( y ){\text d}y\int_0^{x - y} {{e^{ - \lambda t}}{f_S}( t ){\text d}t} }.\label{p12}
\end{align}
This directly indicates that
\begin{align}
P_{i}^{\text{P}}(p_i) =& \Pr\{\Delta_{i}^{\text{P}}>p_i\} =\! 1 \!- F_{\Delta_{i}^{\text{P}}}(x)|_{x=p_i}\nonumber\\
&=1 \!-\frac{1}{{{L_S}( \lambda  )}}\!\int_0^{p_i}\!\!\! {{f_{{Y_i}}}( y ){\text d}y\!\!\int_0^{p_i - y} \!\!{{e^{ - \lambda t}}{f_S}( t ){\text d}t} },
\end{align}
for given PAoI threshold $p_i>0$.
Additionally, noting that ${f_{\Delta _i^{\text{P}}}}(x)=\frac{{{\text{d}}}}{{{\text{d}}x}}{F_{\Delta _i^{\text{P}}}(x)}$, it has that
\begin{align}
{f_{\Delta _i^{\text{P}}}}(x)
=\frac{1}{{{L_S}( \lambda  )}}\int_0^x {{e^{ - \lambda ( {x - y} )}}{f_{{Y_i}}}( y ){f_S}( {x - y} ){\text d}y}.
\end{align}

One should note that age is non-negative, hence, it has ${f_{\Delta _i^{\text{P}}}}(x)=0$ for $x<0$.
This ends the proof.
\subsection{Proof of Theorem \ref{cor p2}}\label{prf p2}
Recalling (\ref{p17}), (\ref{p21}) and (\ref{p11}), one can derive the violation probability of AoI
\begin{align}\label{p23}
&P_i^{\text{A}}(w_i)={\mathbb{E}\left[ {{J_{i,k}}} \right]}/{{{\mathbb{E}\left[ {{Y_i}} \right]}}}\nonumber\\
=& \frac{1}{{{\mathbb{E}\left[ {{Y_i}} \right]}}}\!\left( {\!\int_0^{ + \infty }\!\!\!\! {r{f_{{Y_i}}}\left( r \right)\!{\text d}r \!\!\int_{w_i}^{ + \infty } \!\!\!\!{{f_T}\left( t \right)\!{\text d}t}} } \right. +\!\! \int_0^{w_i}\!\!\!\! {{f_{{Y_i}}}\left( r \right)\!{\text d}r} \!\!\int_{{w_i} - r}^{w_i}\!\!\!\! {\left( {r}\right.} \nonumber\\
&{\left.{r+ t - {w_i}} \right)\!{f_T}\left( t \right)\!{\text d}t}\!\left. { +\!\! \int_{w_i}^{ + \infty }\!\!\!\! {{f_{{Y_i}}}\!\!\left( r \right)\!{\text d}r} \!\!\int_0^{w_i} \!\!\!\!{\left( {r + t - {w_i}} \right)\!{f_T}\left( t \right)\!{\text d}t}\! } \right)\nonumber\\
\mathop  = \limits^{\left( {\text{a}} \right)}& \frac{1}{{{\mathbb{E}\left[ {{Y_i}} \right]}}}\left( {{{{\mathbb{E}\left[ {{Y_i}} \right]}}}\left( {1 - {F_T}\left( {w_i} \right)} \right)} \right.+ {F_T}\left( {w_i} \right)\int_0^{w_i} {r{f_{{Y_i}}}\left( r \right){\text d}r}  \nonumber\\
&- \!\int_0^{w_i}\!\!\! {{f_{{Y_i}}}\left( r \right)\!{\text d}r}\! \int_{{w_i} - r}^{w_i}\!\!\! {{F_T}\left( t \right)\!{\text d}t} + \left. {{F_T}\left( {w_i} \right)\!\int_{w_i}^{ + \infty } \!\!{r{f_{{Y_i}}}\left( r \right)\!{\text d}r}}\right.  \nonumber\\
&\left.{- \int_{w_i}^{ + \infty } {{f_{{Y_i}}}\left( r \right){\text d}r\int_0^{w_i} {{F_T}\left( t \right){\text d}t}} } \right)\nonumber\\
=& 1 - \frac{1}{{{\mathbb{E}\left[ {{Y_i}} \right]}}}\left( {\int_0^{w_i} {{f_{{Y_i}}}\left( r \right){\text d}r} \int_{{w_i} - r}^{w_i} {{F_T}\left( t \right){\text d}t} } \right.\nonumber\\
&\left. { + \int_{w_i}^{ + \infty } {{f_{{Y_i}}}\left( r \right){\text d}r} \int_0^{w_i} {{F_T}\left( t \right){\text d}t} } \right),
\end{align}
where ${{F_T}( t )}$ is the c.d.f. of system time $T$, and the equality (a) holds following from
$\int_{w_i - r}^{w_i} {( {r + t - w_i} ){f_T}( t ){\text d}t}  = {F_T}( w_i )r - \int_{w_i - r}^{w_i} {{F_T}( t ){\text d}t}$
and $\int_0^{w_i} {( {r + t - w_i} ){f_T}( t ){\text d}t}  = {F_T}( w_i )r - \int_0^{w_i} {{F_T}( t ){\text d}t} $.
To obtain a more explicit expression of $P_i^{\text{A}}(w_i)$, we turn to derive the p.d.f. of AoI ${f_{{\Delta _i}}}( w_i )$.

Note that ${f_{{\Delta _i}}}( w_i ) = \frac{\text{d}}{{{\text{d}}w_i}}{F_{{\Delta _i}}}(w_i)=- \frac{\text{d}}{{{\text{d}}w_i}}{P_i^{\text{A}}({w_i})}$. Based on  (\ref{p23}), one can get
\begin{align}\label{p28}
&{f_{{\Delta _i}}}\left( {w_i} \right)=  - \frac{\text{d}}{{{\text{d}}{w_i}}}P_i^{\text{A}}({w_i})\nonumber\\
=& \frac{1}{{\mathbb{E}\left[ {{Y_i}} \right]}}\left( {{f_{{Y_i}}}\left( {w_i} \right)\int_0^{w_i} {{F_T}\left( t \right){\text d}t} } \right.\nonumber\\
& + \int_0^{w_i} {{f_{{Y_i}}}\left( r \right)\left( {{F_T}\left( {w_i} \right) - {F_T}\left( {{w_i} - r} \right)} \right){\text d}r} \nonumber\\
&\left. { - {f_{{Y_i}}}\left( {w_i} \right)\int_0^{w_i} {{F_T}\left( t \right){\text d}t}  + {F_T}\left( {w_i} \right)\int_{w_i}^{ + \infty } {{f_{{Y_i}}}\left( r \right){\text d}r} } \right)\nonumber\\
=& \frac{1}{{\mathbb{E}\left[ {{Y_i}} \right]}}\!\left( {\!{F_T}\left( {w_i} \right)\!\int_0^{ + \infty }\!\!\! {{f_{{Y_i}}}\left( r \right)\!{\text d}r }{- \!\int_0^{w_i} \!\!\!{{f_{{Y_i}}}\left( r \right){F_T}\left( {{w_i} - r} \right)\!{\text d}r} } \!} \right)\nonumber\\
=& \frac{1}{{\mathbb{E}\left[ {{Y_i}} \right]}}\left( {\!{F_T}\left( {w_i} \right) - \int_0^{w_i} {{f_{{Y_i}}}\left( r \right){F_T}\left( {{w_i} - r} \right){\text d}r} } \right).
\end{align}
Let us substitute $w_i$ and $r$ with $x$ and $y$, respectively. \textcolor{black}{Note that ${F_T}( x-y ) = \int_0^{{x-y}} \!\!{{f_T}( t ){\text d}t} =\frac{1}{{{L_S}( \lambda  )}}\int_0^{{x-y}}\!\! {{e^{ - \lambda t}}{f_S}( t ){\text d}t}$, which holds following from (\ref{p8}), and $\mathbb{E}[{Y_i}] =  - { {\frac{{\text{d}}}{{\text{d}s}}}{L_{{Y_i}}( s )} |_{s = 0}} = \frac{1}{{{\lambda _i}{L_S}( \lambda  )}}$.} The p.d.f. of AoI for $x\geq0$ can be obtained. In addition, for $x<0$, we have ${f_{{\Delta _i}}}( x )=0$. Hence, one can get
\begin{align}\label{t28}
{f_{{\Delta _i}}}( x )=\!
\begin{cases}
{\lambda _i}( {\int_0^{{x}}\! {{e^{ - \lambda t}}{f_S}\!\left( t \right){\text d}t}  }\\
{- \int_0^x \!{{f_{{Y_i}}}\!\left( y \right){\text d}y} \int_0^{{{x - y}}}\! {{e^{ - \lambda t}}{f_S}\!\left( t \right){\text d}t} } ), \!\!\!\!\!& x\! \geq \!0,\\
0,\!\!\! \!\!& {x \!<\! 0.}
\end{cases}
\end{align}

Since $f_{Y_i}(x)$ is related to $f_S(x)$ and $\lambda_i$, one can express ${f_{{\Delta _i}}}( x )$ more explicitly with some manipulations. In the following, let us simplify (\ref{t28}) by using the LT of ${f_{{\Delta _i}}}( x )$, which is denoted by $L_{{{\Delta _i}}}(s)$.
First, let us define
\begin{align}
u\left( x \right) =
\begin{cases}
{1},&{x \geq 0,}\\
{0},&{x < 0.}
\end{cases}
\end{align}
Accordingly, one can get $f_{S}(x)=f_{S}(x)u(x)$ and $f_{Y_i}(x)=f_{Y_i}(x)u(x)$.
Therefore, (\ref{t28}) can be rewritten as
\begin{align}\label{t32}
&{f_{{\Delta _i}}}\left( x \right)\nonumber\\
=& {\lambda _i}\left( {u\left( x \right)\int_0^{{x}} {{e^{ - \lambda t}}{f_S}\left( t \right){\text d}t} } \right.\nonumber\\
&\left. { - u\left( x \right)\int_0^x {{f_{{Y_i}}}\left( y \right)} u\left( {x - y} \right)\int_0^{{{x - y}}} {{e^{ - \lambda t}}{f_S}\left( t \right){\text d}t} {\text d}y} \right)\nonumber\\
=& {\lambda _i}\left( {u\left( x \right)\int_0^{{x}} {{e^{ - \lambda t}}{f_S}\left( t \right){\text d}t} } \right.\nonumber\\
&\left. { - u\left( x \right)\int_0^{ + \infty } {{f_{{Y_i}}}\left( y \right)} u\left( {x - y} \right)\int_0^{{{x - y}}} {{e^{ - \lambda t}}{f_S}\left( t \right){\text d}t} {\text d}y} \right)\nonumber\\
=& {\lambda _i}\!\left( {\!\!u\!\left( x \right)\!\!\int_0^{{x}}\!\!\!\! {{e^{ - \!\lambda t}}\!{f_S}\!\left( t \right)\!{\text d}t }{- {f_{{Y_i}}}\!\left( x \right)\! * \!\left( {\!\!u\!\left( x \right)\!\!\int_0^{{x}}\!\!\!\! {{e^{ - \!\lambda t}}\!{f_S}\!\left( t \right)\!{\text d}t} \!} \right)} \!\!} \right)\!\!,
\end{align}
where the operator $ * $ represents the convolution, i.e., for two functions $h_1(x)$ and $h_2(x)$, ${h_1}( x )*{h_2}( x ) = \int_{ - \infty }^{ + \infty } {{h_1}( y ){h_2}( {x - y} ){\text d}y} $.
Based on (\ref{ly}), (\ref{t32}) and the convolution theorem of LT, one can get
\begin{align} \label{pppppp}
{L_{{\Delta _i}}}\left( s \right) &= {\lambda _i}\left( {\frac{{{L_S}\left( {\lambda  + s} \right)}}{s} - {L_{{Y_i}}}\left( s \right)\frac{{{L_S}\left( {\lambda  + s} \right)}}{s}} \right)\nonumber\\
& = \frac{{{\lambda _i}{L_S}\left( {\lambda  + s} \right)}}{{{\lambda _i}{L_S}\left( {\lambda  + s} \right) + s}}.
\end{align}
Hence, the p.d.f. of AoI can be explicitly given by
\begin{align}\label{uuu}
{f_{{\Delta _i}}}\left( x \right)= L^{-1}\left[\frac{{{\lambda _i}{L_S}\left( {\lambda  + s} \right)}}{{{\lambda _i}{L_S}\left( {\lambda  + s} \right) + s}}\right].
\end{align}

Finally, for the given AoI threshold $w_i>0$, one can get the violation probability of AoI
\begin{align}
P_{i}^{\text{A}}(w_i)&=\Pr\{\Delta_{i}^{\text{A}}>w_i\} = 1 - \int_0^{w_i} {{f_{{\Delta _i}}}\left( x \right){\text d}x},
\end{align}
where ${f_{{\Delta _i}}(x)}$ is given by (\ref{uuu}).
This ends the proof.
\subsection{Proof of Lemma \ref{lem 1}}\label{prf lemma1}
First, based on (\ref{p30}), one can get
\begin{align}
{L_{{Y_i}}}\left( s \right)&= \frac{{{\lambda _i}{L_S}\left( {\lambda  + s} \right)}}{{{\lambda _i}{L_S}\left( {\lambda  + s} \right) + s}}= \frac{{{\lambda _i}\mu }}{{{s^2} + \left( {\lambda  + \mu } \right)s + {\lambda _i}\mu }}.
\end{align}
To use the inverse LT to get $f_{Y_i}(x)$, we rewrite ${L_{{Y_i}}}\left( s \right) $ as
\begin{align}\label{p33}
{L_{{Y_i}}}\left( s \right) = \frac{A_i}{{s - a_i}} + \frac{B_i}{{s - b_i}},
\end{align}
where $a_i$ and $b_i$ are the solutions of the quadratic equation of $s$, i.e., ${{s^2} + \left( {\lambda  + \mu } \right)s + {\lambda _i}\mu }=0$. Thus, $A_i$ and $B_i$ satisfy
\begin{align}\label{p34}
\begin{cases}
A_i + B_i = 0,\\
A_ib_i + B_ia_i =  - {\lambda _i}\mu.
\end{cases}
\end{align}
By solving (\ref{p34}), we have $A_i = {{{\lambda _i}\mu }}/({{a_i - b_i}})$ and $B_i =  - {{{\lambda _i}\mu }}/({{a_i - b_i}})$.
Therefore, one can get
\begin{align}
{L_{{Y_i}}}\left( s \right) = \frac{{{\lambda _i}\mu }}{{a_i - b_i}}\left( {\frac{1}{{s - a_i}} - \frac{1}{{s - b_i}}} \right),
\end{align}
which can be used to directly obtain $f_{Y_i}(x)$ via inverse LT.
This ends the proof.
\subsection{Proof of Corollary \ref{cor p3}}\label{prf corollaty3}
To obtain the violation probability of PAoI, we first derive the p.d.f. of PAoI.
Based on Theorem \ref{thm pdfp}, Lemma \ref{lem 1} and (\ref{p29}), for $x\geq0$, one can get
\begin{align}
&{f_{\Delta _i^{\text{P}}}}(x)=\frac{1}{{{L_S}\left( \lambda  \right)}}\int_0^x {{e^{ - \lambda \left( {x - y} \right)}}{f_{{Y_i}}}\left( y \right){f_S}\left( {x - y} \right){\text d}y} \nonumber\\
&= \frac{{\lambda  + \mu }}{\mu }\int_0^x {\frac{{{\lambda _i}\mu }}{{a_i - b_i}}\left( {{e^{a_iy}} - {e^{b_iy}}} \right)\mu {e^{ - \left( {\lambda  + \mu } \right)\left( {x - y} \right)}}{\text d}y}\nonumber\\
&= \frac{{{\lambda _i}\mu }\left( {\lambda  + \mu } \right)}{{a_i - b_i}}{e^{ - \left( {\lambda  + \mu } \right)x}}\int_0^x {{e^{\left( {a_i + \lambda  + \mu } \right)y}} - {e^{\left( {b_i + \lambda  + \mu } \right)y}}{\text d}y}\nonumber\\
&\mathop  = \limits^{\left( {\text{a}} \right)} \frac{{{\lambda _i}\mu }\left( {\lambda  + \mu } \right)}{{a_i - b_i}}{e^{ - \left( {\lambda  + \mu } \right)x}}\int_0^x {{e^{ - b_iy}} - {e^{ - a_iy}}{\text d}y}\nonumber\\
&= \frac{{{\lambda _i}\mu }\left( {\lambda  + \mu } \right)}{{a_i - b_i}}{e^{ - \left( {\lambda  + \mu } \right)x}}( {\frac{{b_i{e^{-a_ix}} - a_i{e^{-b_ix}}}}{{a_ib_i}} + \frac{{a_i - b_i}}{{a_ib_i}}} )\nonumber\\
&\mathop  = \limits^{\left( {\text{b}} \right)} \left( {\lambda  + \mu } \right)( {{e^{ - \left( {\lambda  + \mu } \right)x}} + \frac{{b_i{e^{b_ix}} - a_i{e^{a_ix}}}}{{a_i - b_i}}} ),
\end{align}
in which the equalities (a) and (b) hold  following from (\ref{p37}). Additionally, for $x<0$, we have ${f_{\Delta _i^{\text{P}}}}(x)=0$.

Therefore, the violation probability $P_i^{\text P}(p_i)$ can be given by
\begin{align}
P_i^{\text P}(p_i) = \Pr \left\{ {\Delta _i^{\text P} > p_i} \right\}& = 1 - \int_0^{p_i} {{f_{\Delta _i^{\text{P}}}}(x){\text d}x} \nonumber\\
={e^{ - \left( {\lambda  + \mu } \right)p_i}}& + \frac{{\lambda  + \mu }}{{a_i - b_i}}\left( {{e^{a_ip_i}} - {e^{b_ip_i}}} \right).
\end{align}
This ends the proof.
\subsection{Proof of Proposition \ref{pp1}}\label{prf p3}
First, let us prove that $P_i^{\text P}(\lambda_i;p_i)$ is a convex function w.r.t. $\boldsymbol {\lambda}$. Recall that $a_i$ and $b_i$ are the solutions of ${{s^2} + ( {\lambda  + \mu } )s + {\lambda _i}\mu }=0$, one can obtain $a_i = - \frac{1}{2}( {\lambda  + \mu } ) + \frac{1}{2}{( {{{( {\lambda  + \mu } )}^2} - 4{\lambda _i}\mu })}^{\frac{1}{2}} $ and $b_i = - \frac{1}{2}( {\lambda  + \mu } ) - \frac{1}{2}{( {{{( {\lambda  + \mu } )}^2} - 4{\lambda _i}\mu })}^{\frac{1}{2}} $.
We denote $D_i=a_i-b_i= {( {{{( {\lambda  + \mu } )}^2} - 4{\lambda _i}\mu })}^{\frac{1}{2}}>0$, $D_i'=\frac{{{\text{d}}{D_i}}}{{{\text{d}}{\lambda _i}}}$,  $a_i' = \frac{{{\text{d}}a_i}}{{{\text{d}}{\lambda _i}}}$ and $b_i' = \frac{{{\text{d}}b_i}}{{{\text{d}}{\lambda _i}}}$. To simplify the derivation, we give the following results,
\begin{align}
a_i'\! &=\! \frac{\text{d}}{{{\text{d}}{\lambda _i}}}\!( {- \frac{1}{2}( {\lambda \! + \!\mu} ) + \frac{1}{2} ({{{( {\lambda \! + \!\mu } )}^2}\! - \! 4{\lambda _i}\mu })^{1/2}} )\! = \! - \frac{\mu }{D_i},\label{p55}\\
b_i'\!&=\! \frac{{\text{d}}}{{{\text{d}}{\lambda _i}}}\!( {- \frac{1}{2}( {\lambda  \!+ \!\mu } ) - \frac{1}{2} ({{{( {\lambda \! + \!\mu } )}^2}\! - \! 4{\lambda _i}\mu })^{1/2}} )\! = \!  \frac{\mu }{D_i},\label{p56}\\
D_i'\!&=a_i'-b_i'=- \frac{2\mu }{D_i}.\label{p57}
\end{align}
Therefore, based  on Corollary \ref{thm pdfp} as well as (\ref{p55}), (\ref{p56}) and (\ref{p57}), the first and the second partial derivatives are given by
\begin{align}\label{53}
&\frac{{\partial P_i^{\text P}(\lambda_i;p_i)}}{{\partial {\lambda _i}}}= \frac{\partial }{{\partial {\lambda _i}}}( {{e^{ - ( {\lambda  + \mu } ){p_i}}} + \frac{{\lambda  + \mu }}{{a_i - b_i}}( {{e^{a_i{p_i}}} - {e^{b_i{p_i}}}} )} )\nonumber\\
&= \frac{\lambda  + \mu }{{{{( {a_i - b_i} )}^2}}} ({{( {a_i'{p_i}{e^{a_i{p_i}}} - b_i'{p_i}{e^{b_i{p_i}}}} )( {a_i - b_i} )}}\nonumber\\
&{- ( {a_i' - b_i'} )( {{e^{a_i{p_i}}} - {e^{b_i{p_i}}}} )})\nonumber\\
&=  - \frac{ {\lambda \! +\! \mu }}{{{D_i^2}}}({{D_i( {\frac{\mu }{D_i}{p_i}{e^{b_i{p_i}}} + \!\frac{\mu }{D_i}{p_i}{e^{a_i{p_i}}}} ) }{+  \frac{{2\mu }}{D_i}( {{e^{b_i{p_i}}} - {e^{a_i{p_i}}}} )}})\nonumber\\
&=  - \frac{{( {\lambda  + \mu } )\mu }}{{{D_i^3}}}( {( {D_i{p_i} + 2} ){e^{b_i{p_i}}} + ( {D_i{p_i} - 2} ){e^{a_i{p_i}}}} ),\\
&\frac{{\partial P_i^{\text P}(\lambda_i;p_i)}}{{\partial {\lambda _j}}} = 0,\label{54}\\
&\frac{{{\partial ^2}P_i^{\text{P}}(\lambda_i;p_i)}}{{\partial {\lambda _i}^2}}
=  - \frac{( {\lambda  + \mu } )\mu}{{D_i^6}}({( {( {D_i'{p_i} + b_i'{p_i}( {D_i{p_i} + 2} )} ){e^{b_i{p_i}}}} }\nonumber\\
&+{{( {D_i'{p_i} + a_i'{p_i}( {D_i{p_i} - 2} )} ){e^{a_i{p_i}}}} ){D_i^3}} \nonumber\\
&{- 3D_i'{D_i^2}( {( {D_i{p_i} + 2} ){e^{b_i{p_i}}} + ( {D_i{p_i} - 2} ){e^{a_i{p_i}}}} )})\nonumber\\
&= \frac{{( {\lambda  + \mu } ){\mu ^2}}}{{{D_i^5}}}( {( {{D_i^2}{p_i}^2 - 6D_i{p_i} + 12} ){e^{a_i{p_i}}}}\nonumber\\
&{- ( {{D_i^2}{p_i}^2 + 6D_i{p_i} + 12} ){e^{b_i{p_i}}}} )\nonumber\\
&= \frac{{( {\lambda  + \mu } ){\mu ^2}}}{{{D_i^5}}}{e^{b_i{p_i}}}( {( {{D_i^2}{p_i}^2 - 6D_i{p_i} + 12} ){e^{D_i{p_i}}} }\nonumber\\
&{- ( {{D_i^2}{p_i}^2 + 6D_i{p_i} + 12} )} ),\label{p60}\\
&\frac{{{\partial ^2}P_i^{\text P}(\lambda_i;p_i)}}{{\partial {\lambda _j}^2}}=\frac{{{\partial ^2}P_i^{\text P}(\lambda_i;p_i)}}{{\partial {\lambda _i}{\lambda _j}}} = \frac{{{\partial ^2}P_i^{\text P}(\lambda_i;p_i)}}{{\partial {\lambda _j}{\lambda _i}}}=0,\label{p61}
\end{align}
for all $ j\neq i$.

To prove that $P_i^{\text P}(\lambda_i;p_i)$ is a convex function w.r.t. $\boldsymbol {\lambda}$, we need to prove the corresponding Hessian Matrix is semi-positive definite. Based on (\ref{p60}) and (\ref{p61}), it can be found that we only need to prove that $\frac{{{\partial ^2}P_i^{\text P}(\lambda_i;p_i)}}{{\partial {\lambda _i}^2}}$
 is positive for $p_i>0$. To this end, let us define
\begin{align}
g( x ) := ( {{x^2} - 6x + 12} ){e^x} - ( {{x^2} + 6x + 12} ).\nonumber
\end{align}
One can obtain the first, the second and the third derivatives
\begin{align}
g'( x ) &= ( {{x^2} - 4x + 6} ){e^x} - ( {2x + 6} ),\nonumber\\
g''( x ) &= ( {{x^2} - 2x + 2} ){e^x} - 2,\nonumber\\
g'''( x )&= {x^2}{e^x} \ge 0, {\text{ for }} x\geq0.\nonumber
\end{align}
Since $g( 0 ) = 0$, $g'( 0 ) = 0$ and $g''( 0 ) = 0$, one can get $g''( x )>0$, $g'( x )>0$ and $g( x )>0$, for $x>0$.
Accordingly, based on (\ref{p60}) and $p_i>0$, one can get
\begin{align}
\frac{{{\partial ^2}P_i^{\text P}(\lambda_i;p_i)}}{{\partial {\lambda _i}^2}}
=\frac{{( {\lambda  + \mu } ){\mu ^2}}}{{{D_i^5}}}{e^{b_i{p_i}}}g( {D_i{p_i}} )>0.
\end{align}
Thus, $P_i^{\text P}(\lambda_i;p_i)$ is a convex function w.r.t. $\boldsymbol {\lambda}$.

Because of the convexity-preserving property of maximum operator, one can get $P_{\text M}^{\text P}(\boldsymbol {\lambda};\boldsymbol {p})$ is also a convex function w.r.t. $\boldsymbol {\lambda}$. Furthermore, noting that the feasible set of ${\mathcal{P}_1 }$ is convex, it can be concluded that $\mathcal{P}_1$ is a convex problem.

Finally, let us prove that $P_i^{\text{P}}(\lambda_i;p_i)$ monotonically decreases with $\lambda_i$, regardless of  $\lambda_j$, for all $j\neq i$.
To obtain the monotonicity, it is only  needed to prove $\frac{{\partial P_i^{\text P}(\lambda_i;p_i)}}{{\partial {\lambda _i}}}=- \frac{{( {\lambda  + \mu } )\mu }}{{{D_i^3}}}( {( {D_i{p_i} + 2} ){e^{b_i{p_i}}} + ( {D_i{p_i} - 2} ){e^{a_i{p_i}}}} )= - \frac{{( {\lambda  + \mu } )\mu }}{{D_i^3}}{e^{{b_i}{p_i}}}( {( {{D_i}{p_i} + 2} ) + ( {{D_i}{p_i} - 2} ){e^{{D_i}{p_i}}}} )>0$. Let us define a function $h(x):=(x-2)e^x+x+2$, whose first and second  derivatives are $h'(x)=(x-1)e^x+1$ and $h''(x)=xe^x\geq0$, for $x\geq0$. One can get $h'(0)=0$ and $h(0)=0$, which imply that $h'(x)>0$ and $h(x)>0$ when $x>0$. Hence, it can be obtained that $\frac{{\partial P_i^{\text P}(\lambda_i;p_i)}}{{\partial {\lambda _i}}}= - \frac{{( {\lambda  + \mu } )\mu }}{{D_i^3}}{e^{{b_i}{p_i}}}h({ D_ip_i})<0$, i.e., $P_i^{\text{P}}(\lambda_i;p_i)$ monotonically decreases with $\lambda_i$. Noting that $\frac{{\partial P_i^{\text P}(\lambda_i;p_i)}}{{\partial {\lambda _j}}} =0$, one can directly obtain that $P_i^{\text P}(\lambda_i;p_i)$ is uncorrelated with  $\lambda_j$, for all $j\neq i$.
This ends the proof.
\subsection{Proof of Proposition \ref{pp2}}\label{prf p4}
The proof of Proposition \ref{pp2} is similar to that of Proposition \ref{pp1}, as shown in the following.
First, let us prove that $P_i^{\text A}(\lambda_i;w_i)$ is a convex function w.r.t. $\boldsymbol {\lambda}$.
Based  on Corollary \ref{cor p5} as well as (\ref{p55}), (\ref{p56}) and (\ref{p57}), the first and the second partial derivatives of $P_i^{\text A}$ w.r.t. $(\lambda_1,\lambda_2,\ldots,\lambda_N )$ are given by
\begin{align}
&\frac{{\partial P_i^{\text A}(\lambda_i;w_i)}}{{\partial {\lambda _i}}}= \frac{\partial }{{\partial {\lambda _i}}}( {\frac{{a_i{e^{b_i{w_i}}} - b_i{e^{a_i{w_i}}}}}{{a_i - b_i}}} )\nonumber\\
&= \frac{1}{{{{( {a - b} )}^2}}}({{( {a_i'{e^{b_i{w_i}}} + a_ib_i'{w_i}{e^{b_i{w_i}}} - b_i'{e^{a_i{w_i}}}}}} \nonumber\\
&{{{- b_ia_i'{p_i}{e^{a_i{w_i}}}} )( {a_i - b_i} ) - ( {a_i' - b_i'} )( {a{e^{b{w_i}}} - b{e^{a{w_i}}}} )}})\nonumber\\
&= \frac{1}{{{D_i^2}}}({{( { - \frac{\mu }{D_i}{e^{b_i{w_i}}} + a_i\frac{\mu }{D_i}{w_i}{e^{b_i{w_i}}} - \frac{\mu }{D_i}{e^{a_i{w_i}}}}}}\nonumber\\
&{{{+ b_i\frac{\mu }{D_i}{w_i}{e^{a_i{w_i}}}} )D_i + \frac{{2\mu }}{D_i}( {a_i{e^{b_i{w_i}}} - b_i{e^{a_i{w_i}}}} )}})\nonumber\\
&=\! \frac{\mu}{{D_i^3}}\! \left({\!{\left( {a_iD_i{w_i} \!- \!D_i \!+ \!2a_i} \right)\!{e^{b_i{w_i}}}\! + \!\left( {b_iD_i{w_i}\! -\! D_i \!-\! 2b_i} \right)\!{e^{a_i{w_i}}}}\!}\right)\nonumber\\
&\mathop= \limits^{\left( {\text{a}} \right)}\!\frac{\mu}{{D_i^3}}\! \left({\!\left( {\!a_iD_i{w_i}\! -\! \left( {\!\lambda \! + \!\mu \!} \right)\!} \right)\!{e^{b_i{w_i}}}\! + \! \left( {\!b_iD_i{w_i}\!+ \!\left( {\!\lambda \! + \!\mu \!} \right)\!} \right)\!{e^{a_i{w_i}}}\!}\right)\!\!,\label{62}\\
&\frac{{\partial P_i^{\text A}}}{{\partial {\lambda _j}}} = 0,\label{63}\\
&\frac{{{\partial ^2}P_i^{\text A}(\lambda_i;w_i)}}{{\partial {\lambda _i}^2}}\nonumber\\
&= \frac{\mu }{{{D_i^6}}}( {{D_i^3}( {( {a_i'D_i{w_i} + a_iD_i'{w_i} + b_i'{w_i}( {a_iD_i{w_i} - ( {\lambda }}}}} \nonumber\\
&{{{{{+ \mu } )} )} ){e^{b{w_i}}} + ( {b_i'D_i{w_i} + b_iD_i'{w_i} + a_i'{w_i}( {b_iD_i{w_i} + ( {\lambda  } }  } } } \nonumber\\
&{{{{{+ \mu } )} )} ){e^{a_i{w_i}}}} ) - 3D_i'{D_i^2}( {( {a_iD_i{w_i} - ( {\lambda  + \mu } )} ){e^{b_i{w_i}}} }}\nonumber\\
&{{+ ( {b_iD_i{w_i} + ( {\lambda  + \mu } )} ){e^{a_i{w_i}}}} )} )\nonumber\\
&= \frac{{{\mu ^2}}}{{{D_i^5}}}({( {a_i{{( {D_i{w_i}} )}^2} + ( {4a_i + 2b_i} )D{w_i} - 6( {\lambda  + \mu } )} ){e^{b_i{w_i}}}} \nonumber\\
&+{( { - b_i{{( {D_i{w_i}} )}^2} + ( {4b_i + 2a_i} )D_i{w_i} + 6( {\lambda  + \mu } )} ){e^{a_i{w_i}}}\!})\nonumber\\
&\mathop  =\limits^{( {\text{b}} )}  \frac{{{\mu ^2}}}{{{D_i^5}}}{e^{b_i{w_i}}}( {( { - b_i{{( {D_i{w_i}} )}^2} + ( {4b_i + 2a_i} )D_i{w_i} - 6( a_i+b_i )} )}\nonumber\\
&\times{{e^{D_i{w_i}}} \!+\! a_i{{( {D_i{w_i}} )}^2} \!+ \!( {4a_i \!+ \!2b_i} )D_i{w_i} \!+ 6( a_i+b_i )} ),\label{p70}\\
&\frac{{{\partial ^2}P_i^{\text{A}}}(\lambda_i;w_i)}{{\partial {\lambda _j}^2}}=\frac{{{\partial ^2}P_i^{\text{A}}(\lambda_i;w_i)}}{{\partial {\lambda _i}{\lambda _j}}} = \frac{{{\partial ^2}P_i^{\text{A}}(\lambda_i;w_i)}}{{\partial {\lambda _j}{\lambda _i}}}=0,  \label{p71}
\end{align}
for all $j\neq i$, where equalities (a) and (b) hold following from (\ref{p37}).

To prove that $P_i^{\text{A}}(\lambda_i;w_i)$ is a convex function w.r.t. $\boldsymbol {\lambda}$, we need to prove the corresponding Hessian Matrix is semi-positive definite. Based on (\ref{p70}) and (\ref{p71}), it can be found that we only need to prove that $\frac{{{\partial ^2}P_i^{\text{A}}(\lambda_i;w_i)}}{{\partial {\lambda _i}^2}}$
 is positive for $w_i>0$. To this end, let us define
\begin{align}
f_i\left( x \right): = &\left( { - b_i{x^2} + \left( {4b_i + 2a_i} \right)x - 6\left( {a_i + b_i} \right)} \right){e^x} \nonumber\\
&+ a_i{x^2} + \left( {4a_i + 2b_i} \right)x + 6\left( {a_i + b_i} \right). \nonumber
\end{align}
One can obtain the first, the second and the third derivatives
\begin{align}
f_i' \!\left(  x  \right) \! &= \! \left( {  - b_i{x^2} \! + \! 2\left( {\! a_i \! +  b_i \!} \right) \!x \! - \! 2b_i \! - \! 4a_i  } \right) \!{e^x}  \!\! + \! 2a_ix  \!+ \! 4a_i \! + \! 2b_i,\nonumber\\
f_i'' \!\left(  x  \right) \!&=  \!\left( { - b_i{x^2} + 2a_ix - 2a_i} \right){e^x} + 2a_i,\nonumber\\
f_i''' \!\left(  x  \right) \!&=  \!\left( { - b_i{x^2} + 2\left( {a_i - b_i} \right)x} \right){e^x}\mathop  \ge  0,\nonumber
\end{align}
for $x\geq0$, where the inequality holds  following from $b_i <0$ and $a_i-b_i>0$.
Since $f_i( 0 ) = 0$, $f_i'( 0 ) = 0$, $f_i''( 0 ) = 0$, one can get $f_i''( x )>0$, $f_i'( x )>0$ and $f_i( x )>0$, when $x>0$.
Accordingly, based on (\ref{p70}), one can get
\begin{align}
\frac{{{\partial ^2}P_i^{\text {A}}}(\lambda_i;w_i)}{{\partial {\lambda _i}^2}}
 = \frac{{{\mu ^2}}}{{{D_i^5}}}{e^{b_i{w_i}}}f_i\left( {D_i{w_i}} \right)>0.
\end{align}
Hence, $P_i^{\text{A}}(\lambda_i;w_i)$ is a convex function w.r.t. $\boldsymbol {\lambda}$.

Because of the convexity-preserving property of maximum operator, one can get $P_{\text{M}}^{\text{A}}(\boldsymbol {\lambda};\boldsymbol {w})$ is also a convex function w.r.t. $\boldsymbol {\lambda}$. Furthermore, noting the convexity of the feasible set, it can be concluded that $\mathcal{P}_2$ is a convex problem.

Finally, let us prove that $P_i^{\text{A}}(\lambda_i;w_i)$ monotonically decreases with $\lambda_i$, regardless  of $\lambda_j$, for all $j\neq i$.
To prove the monotonicity, it can be found that we only need to prove $\frac{{\partial P_i^{\rm{A}}(\lambda_i;w_i)}}{{\partial {\lambda _i}}}< 0$. Similarly, let us define $r_i( x ) := ( {{b_i}x + ( {\lambda  + \mu } )} ){e^x} + {a_i}x - ( {\lambda  + \mu } )$. According to (\ref{p37}), $r_i(x)$'s first and second derivatives, i.e. $r_i'(x)$ and $r_i''(x)$, can be obtained.
One can get $r_i'(0)=0$, $r_i(0)=0$ and $r_i''(x)<0$ for $x\geq0$, which imply that $r_i'(x)<0$ and $r_i(x)<0$ for $x>0$. Thus, it can be concluded that $\frac{\partial P_i^{\rm{A}}(\lambda_i;w_i)}{{\partial {\lambda _i}}} = \frac{\mu }{{D_i^3}}{e^{{b_i}{w_i}}}r_i( {{D_i}{w_i}} ) < 0$.
Noting that $\frac{{\partial P_i^{\rm{A}}(\lambda_i;w_i)}}{{\partial {\lambda _j}}} =0$, one can directly obtain that $P_i^{\rm{A}}(\lambda_i;w_i)$ is uncorrelated with  $\lambda_j$, for all $j\neq i$.
This ends the proof.

\baselineskip=18pt
\bibliographystyle{IEEEtran}
\bibliography{vpp}

\vspace{-11 mm}
\begin{IEEEbiography}[{\includegraphics[width=1in,height=1.25in,clip,keepaspectratio]{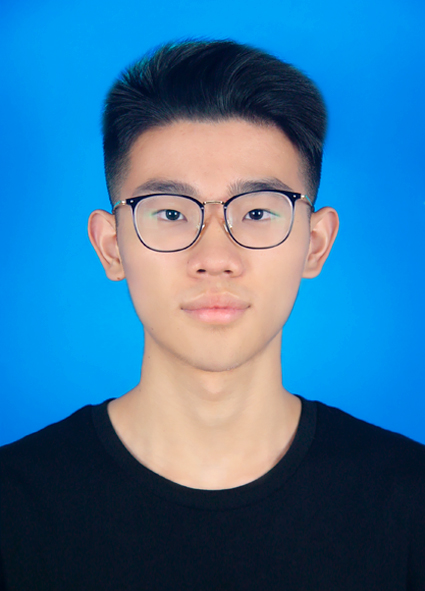}}]{Tianci Zhang}
received the B.S. degree from Chongqing University, China, in 2022, where he is currently pursuing the M.S. degree in information and communication engineering.
His research interests include age of information, wireless communications and wireless networking.
\end{IEEEbiography}
\vspace{-11 mm}
\begin{IEEEbiography}[{\includegraphics[width=1in,height=1.25in,clip,keepaspectratio]{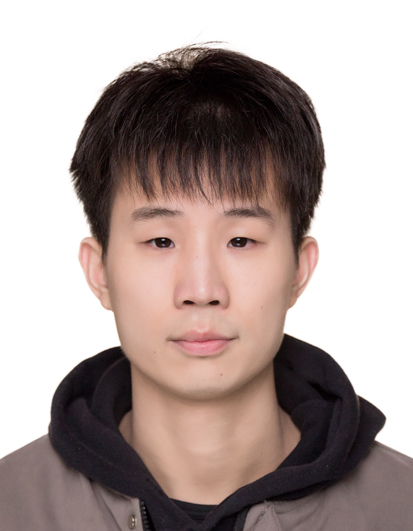}}]{Shutong Chen}
received the B.S. degree from Chongqing University, Chongqing, China, in 2022. He is currently pursuing the M.S. degree in University College London, London WC1E 6BT, U.K. His research interests include performance evaluations and age of information.
\end{IEEEbiography}
\vspace{-11 mm}
\begin{IEEEbiography}[{\includegraphics[width=1in,height=1.25in,clip,keepaspectratio]{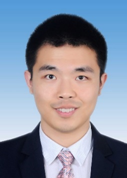}}]{Zhengchuan Chen}
received the Ph.D. degree from Tsinghua University, China, in 2015. He is currently an Assistant Professor with the School of Microelectronics and Communication Engineering, Chongqing University, China. His research interests include reconfigurable intelligent surface, age of information, and network information theory. He has served as a Technical Program Committee Member for several IEEE conferences. He received the Best Paper Award at the International Workshop on High Mobility Wireless Communications in 2013.
\end{IEEEbiography}
\vspace{-11 mm}
\begin{IEEEbiography}[{\includegraphics[width=1in,height=1.25in,clip,keepaspectratio]{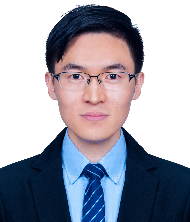}}]{Zhong Tian}
is an assistant professor in the School of Microelectronics and Communication Engineering, Chongqing University, China. He has earned his Ph.D. degrees at Tsinghua University, China, in 2019. His research interests include cognitive radios networks, non-orthogonal multiple access, the blockchain-based spectrum sharing and the reconfigurable intelligent surface.
\end{IEEEbiography}
\vspace{-11 mm}
\begin{IEEEbiography}[{\includegraphics[width=1in,height=1.25in,clip,keepaspectratio]{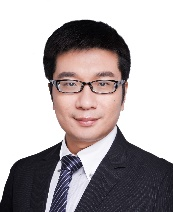}}]{Yunjian Jia}
received the Ph.D. degree from Osaka University, Japan, in 2006. From 2006 to 2012, he was a researcher with Central Research Laboratory, Hitachi, Ltd., where he engaged in research and development on wireless networks, and contributed to LTE/LTE-Advanced standardization in 3GPP. He is now a professor at the School of Microelectronics and Communication Engineering, Chongqing University, China. He is the author of more than 70 published papers and the inventor of 32 patents.
\end{IEEEbiography}
\vspace{-11 mm}
\begin{IEEEbiography}[{\includegraphics[width=1in,height=1.25in,clip,keepaspectratio]{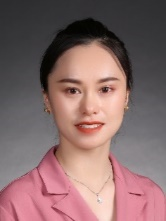}}]{Min Wang}
received the Ph.D. degree from the Department of Electronic Engineering, Tsinghua University, Beijing, China, in 2019. Since 2019, she has been with the School of Optoelectronics Engineering, Chongqing University of Posts and Telecommunications, Chongqing. Dr. Wang's current research interests include transmitarray and reflectarray antennas, reconfigurable antennas, millimeter-wave antennas, arrays and metasurfaces analysis.
\end{IEEEbiography}
\vspace{-11 mm}
\begin{IEEEbiography}[{\includegraphics[width=1in,height=1.25in,clip,keepaspectratio]{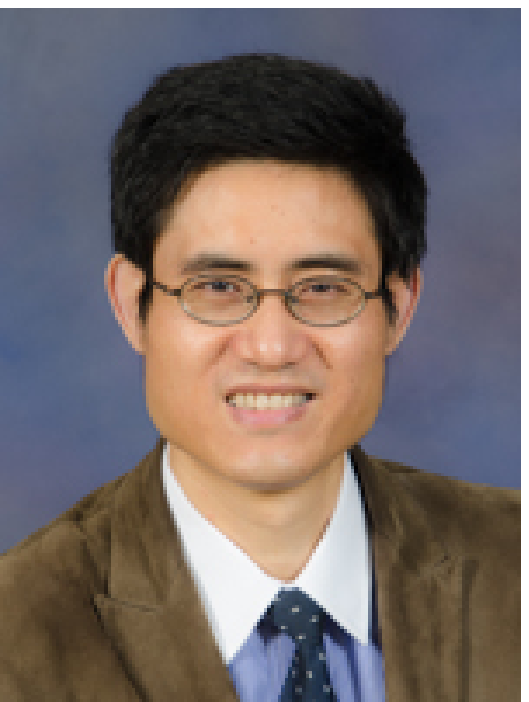}}]{Dapeng Oliver Wu}
(Fellow, IEEE) received the
Ph.D. degree in electrical and computer engineering
from Carnegie Mellon University, Pittsburgh, PA,
USA, in 2003. He is currently a Professor with the Department of Computer Science, City University of HongKong, HongKong. His research
interests include networking, communications, signal
processing, computer vision, machine learning, smart
grid, and information and network security.


\end{IEEEbiography}

\ifCLASSOPTIONcaptionsoff
  \newpage
\fi
\end{document}